\documentclass[letterpaper]{article} % DO NOT CHANGE THIS
\usepackage{aaai2026}  % DO NOT CHANGE THIS

\usepackage{times}  % DO NOT CHANGE THIS
\usepackage{helvet}  % DO NOT CHANGE THIS
\usepackage{courier}  % DO NOT CHANGE THIS
\usepackage[hyphens]{url}  % DO NOT CHANGE THIS
\usepackage{amsthm}
\usepackage{graphicx} % DO NOT CHANGE THIS
\usepackage{natbib}  % DO NOT CHANGE THIS AND DO NOT ADD ANY OPTIONS TO IT
\usepackage{caption} % DO NOT CHANGE THIS AND DO NOT ADD ANY OPTIONS TO IT
\usepackage{algorithm}
\usepackage{algorithmic}
\usepackage{pgfplots}  
\usepackage{amsmath}
\pgfplotsset{width=10cm,compat=newest}
\usepgfplotslibrary{units}
\usetikzlibrary{spy,backgrounds}
\usepackage{pgfplotstable}
\def\eps{\varepsilon}
\usepackage{tikz}
 \usetikzlibrary{shapes.geometric, arrows, fit}
 \usepackage{stmaryrd}
 \usetikzlibrary{decorations.pathreplacing}
\usepackage{multirow}
\usepackage{pgfplots}   
\pgfplotsset{width=10cm,compat=newest}
\usepgfplotslibrary{units}
\usetikzlibrary{spy,backgrounds}
\usepackage{pgfplotstable}
\pgfplotsset{width=10cm,compat=1.9}
\def\eps{\varepsilon}

\newcount\Comments  % 0 suppresses notes to selves in text
\usepackage{pgfplots}
\pgfplotsset{compat=1.18}
\usepackage{subcaption} % For subfigures
\definecolor{darkgreen}{rgb}{0,0.6,0}
\newcommand{\kibitz}[2]{\ifnum\Comments=1{\textcolor{#1}{{#2}}}\fi}
\newcommand{\rmr}[1]{\kibitz{red}{[RM:#1]}}
\newcommand{\gan}[1]{\kibitz{blue}{[GG:#1]}}
	    \newcommand{\red}[1]{{\leavevmode\color{red}{#1}}}

\newcommand\todo[1]{{\red{TODO: {#1}}}}

\newcommand\toref{\red{[REF]}}

\newcommand{\ppm}{p}

\newcommand{\newpar}[1]{\vspace{-1mm}\paragraph{#1}}
\def\calD{\mathcal{D}}
\def\calN{\mathcal{N}}

\def\ol{\overline}
\def\ul{\underline}

\newcommand{\newsec}[1]{\vspace{-3mm}\section{#1}}
\newcommand{\newsubsec}[1]{\vspace{-2mm}\subsection{#1}}
\usepackage{newfloat}
\usepackage{listings}
\usepackage{enumitem}

\usepackage{newfloat}
\usepackage{listings}
\DeclareCaptionStyle{ruled}{labelfont=normalfont,labelsep=colon,strut=off} % DO NOT CHANGE THIS
\floatstyle{ruled}
\newfloat{listing}{tb}{lst}{}
\floatname{listing}{Listing}
\title{On Condorcet's Jury Theorem with  Abstention}
\author{
   Reshef Meir\textsuperscript{\rm 1} and Ganesh Ghalme\textsuperscript{\rm 2} %\thanks{.}
}
\affiliations{
    \textsuperscript{\rm 1}Technion-Israel\\
    \textsuperscript{\rm 2}IITH-India
}

 \newcommand{\floor}[1]{\left \lfloor{#1}\right \rfloor }
\definecolor{ForestGreen}{rgb}{.13,.54,.13}
\definecolor{BrickRed}{rgb}{.80,.26,.33}

\newtheorem{theorem}{Theorem}

\newtheorem{observation}[theorem]{Observation}
\newtheorem{example}{Example}

 \renewcommand{\P}{\mathbb{P}}

\newtheorem{lemma}[theorem]{Lemma}

\newtheorem{proposition}[theorem]{Proposition}

\newtheorem{definition}{Definition}

\usepackage{bbold}
\usepackage{thm-restate}
\usepackage{tikz}
\DeclareMathOperator*{\Exp}{\mathbb{E}}
  
\begin{document}

\maketitle 
%%%%%%%%%%%%%%%%%%%%%%%%%%%%%%%%%%%%%%%%%%%%%%%%%%%%%%%%%%%%%%%%%%%%%%%%%

% Include a short abstract here (100-300 words):
\begin{abstract}

The well-known Condorcet Jury Theorem states that, under majority rule, the better of two alternatives is chosen with probability approaching one as the population grows. We study an asymmetric setting where voters face varying participation costs and share a possibly heuristic belief about their pivotality (ability to influence the outcome).

In a costly voting setup where voters abstain if their participation cost is greater than their pivotality estimate, we identify a single property of the heuristic belief---weakly vanishing pivotality---that  gives rise to multiple stable equilibria in which elections are nearly tied. In contrast, strongly vanishing pivotality (as in the standard Calculus of Voting model) yields a unique, trivial equilibrium where only zero-cost voters participate as the population grows.  We then characterize when nontrivial equilibria satisfy  a version of the Jury Theorem: below a sharp threshold, the majority-preferred candidate wins with probability approaching one; above it, both candidates either win with equal probability. % or maintain a constant winning chance, independent of population size or participation cost distribution.
%We identify a single property of the pivotality estimation function, called weakly vanishing pivotality, which  may give rise to multiple stable equilibria where candidates are nearly-tied in the induced game. 
%In contrast, \emph{strongly} vanishing pivotality estimation functions (which include the `standard' Calculus of Voting  from the literature), lead to a unique trivial equilibrium, in which only 0-cost voters vote. We then characterize the conditions under which nontrivial equilibria admit a `Jury theorem' as population grows: there is a sharp threshold below which majority opinion is selected with a probability that approaches 1; whereas above the threshold both candidates are equally likely to win or, in contrast, both candidates can win with at least a constant probability that does not depend on the size of the population, nor on the distribution of participation costs.
\end{abstract}

\newsec{Introduction}
%\gan{background: Condorcet's result.. }
Consider a population of $N$ voters voting over two alternatives  $A$ and $B$,  with $A$ being the {\em better alternative}  according to some pre-defined criterion. Consider further that the preference of each individual voter is determined independently by an outcome of a coin toss biased in favour of the better alternative $A$. That is, each individual voter supports alternative $A$ with probability $p$ and $B$  with the probability $1-p$. 
Under this setting, the famous Condorcet's Jury Theorem (CJT) states that the majority rule selects candidate $A$ with certainty when population size increases to infinity.  

%Consider a population voting over two alternatives or candidates $A$ and $B$, with $A$ being the `better' alternative according to some objective. If every individual has a probability $p>0.5$ of favouring the better candidate, the famous Condorcet's Jury theorem entails that the majority outcome will select $A$ with probability that tends to 1 as the population grows. 

%\gan{gap/problem: condorcet result assumes everyone votes...this is not the case in practice...presidential elections have around 60\% turnout}

An implicit assumption in Condorcet's  theorem is that \emph{everyone votes}, or at least that the decision  to vote 
 is independent of one's preference over alternatives. 
In contrast, in many practical situations such as  political elections, or a local or national referendum, abstention is found to be a  common and prominent phenomenon. For instance, the voter turnout in US presidential elections has been around 52\%-62\% over the past 90 years~\cite{martinez2005effects}. Abstention is also observed to be a significant phenomenon  in  lab experiments~\cite{Blais,Owen1984}.

%\gan{connect Condorcet's result and  the purpose of this paper }

%In a real world democratic elections, this is rarely true. In fact, a sizable fraction of voter population  \emph{choose} to abstain from voting \cite{riker68, Blais}. The study of why people choose to abstain and which specific fraction  is more likely to abstain is one of the most important question in the study of large elections.

%\gan{ paradox of voting..why people vote...do those who abtain affect the outcome of the election ... what happens when the size of population grows to infinity }
 %From a rational, economic point of view, the surprise is not that some voters abstain, but that anyone votes at all. As Anthony Downs claimed already in 1957,  a rational voter weighs the benefit of voting (which realizes only if the voter is pivotal) against a predetermined cost. When  the  size of the electorate is large, the expected benefit derived from affecting the outcome of the election (i.e. being pivotal) is too small to induce voting  from a significant fraction of voters, thereby giving rise to the `paradox of voting'~\cite{downs57}. This induces  the unique  trivial equilibrium where only the zero cost\footnote{When the population size is large but finite, only voters with arbitrarily small voting cost vote. } voters vote.  
 From a rational, economic perspective, the surprise is not why some voters abstain, but why anyone votes at all. As Anthony Downs noted in 1957, a rational voter compares the benefit of voting—which occurs only if they are pivotal—against a fixed cost. In large electorates, the probability of being pivotal is so low that the expected benefit rarely outweighs the cost. This leads to the so-called paradox of voting~\cite{downs57}, where the only equilibrium is trivial: only voters with zero (or arbitrarily small\footnote{In large but finite populations, only voters with near-zero voting costs participate.}) costs turn out.

%\gan{paradox of voting vs condorcets result... reality is between... does there exists a boundedly rational voting model that explains large turnout and at the same time does Condorcets result holds..strategic abstention }
While Condorcet's result holds under the assumption that everyone votes,  Down's theory of rational voting predicts meagre participation in large-scale elections. %However, the  data from large-scale elections (US presidential elections, for instance)  suggests that  a significant fraction of voters vote,    contradicting the prediction by the  rational voting model in practice. 
%This paper evaluates Condorcet's result in a two-candidate election setting where voters make a strategic, yet bounded-rational choice whether to vote or abstain.  %Our heuristics-based abstention model predicts the constant fraction of voter turnout and shows that Condorcet's result does not always hold even when the population size increases to infinity.    %One of the most important properties  of the proposed abstention model  is  that the voting population is a significant, non-vanishing fraction of the overall population in every induced equilibrium. 
Our goal in this paper is to understand the equilibria that arise from rational voting but with a flexible estimation of one's chances to be pivotal, in an attempt to reconcile the theoretical predictions with the moderate turnout rates we see in practice. We ask: (1) are there equilibria where a significant portion of the population votes? (2) how likely is the better candidate to win, in equilibrium? 
We are mainly interested in the answer as the size of the population grows to infinity: does the winning probability of the better candidate approach 1? Or does the paradox of voting `kill' Condorcet's Jury theorem?   We now present the standard model of voter turnout from the literature, which formalizes how voters estimate their likelihood of being pivotal.

% This theory, however,  does not explain  large turnout observed in practice, leading to the famous \emph{paradox of voting} (a.k.a. Down's paradox). 

%\newsubsec{Theoretical Models of Voter Turnout}
 %We consider  self-interested voters who heuristically estimate the importance of individual votes (perceived pivotality) based on two factors;  {\em number of voters} and {\em margin of victory}.

\label{sec:intro}  
  \newpar{The Calculus of Voting model:}  
  Originally proposed by \citet{downs57} and later developed by 
 \citet{riker68}, this  model  of rational voting attributes each  voter's decision to abstain to the expected cost-benefit analysis. Let $p_i$ denote the \emph{pivotality} of voter~$i$, $\texttt{V}_i$ denote the  personal benefit she receives if her preferred candidate wins an election,   $\texttt{D}_i$ denote the social benefit she receives by performing a civic duty of voting and $\texttt{G}_i$ denote the   costs of voting she incurs. These costs  include the cost of obtaining and processing information and the actual cost of  registering and going to polls (see also~\cite{Aldrich}  for discussion of voting and rational choice).  A voter $i$ votes if and only if %\rmr{we use B for voter type so perhaps use V here (for value)}
  \begin{equation}
     p_i\cdot \texttt{V}_i + \texttt{D}_i \geq \texttt{G}_i. 
     \label{eq:calculus}
 \end{equation} 

 The calculus of voting  model considers  $p_i$ to be the probability that  that all voters except $i$ reach a tie. The tie probabilities are derived from the aggregated stochastic votes, and thus the pivot  computation and subsequent equilibrium analysis quickly become intractable. % as the number of individual voters increases. %We call the calculus of voting model  a  fully rational model. %Apart from computational tractability issues,  another significant problem with fully rational model is that it fails to model voter turnout  in real-world elections accurately. 
 
 %The next followup extensions we describe relax the fully rational  model by replacing the exact calculation of  pivot probability  with approximate values that are only correct in the limit.

%\rmr{Aldrich: there is a paradox in empirical data, where studies based on aggregate data show a strong connection of $p$ and turnout, and studies based on survey data show the opposite. I'm not sure I understand the explanation, or the difference between these studies :(   It may suggest some evidence that voters mis-estimate their pivotality}

%\input{AAAI25/AAAI-NonJury/old_related_work}
%\input{fig_models}

    \newpar{Enter heuristics} Several recent models maintain the fundamental game-theoretic approach  of equilibrium among strategic voters, but relax the assumption that voters calculate their true pivot probabilities, or engage in probabilistic calculations at all. This includes  {minmax regret equilibrium}~\cite{Ferejohn74}, {sampling equilibrium}~\cite{OSBORNE2003434}, 
     or {Local-dominance equilibrium}~\cite{Meir14,meir2015plurality}. \citet{Merrill1981}  considers   (as we also do later) voters maximizing expected utility but without specifying how they estimate their pivot probability.
     \newpar{Exit rationality} Finally, there are  models that suggest voting heuristics people may use, without engaging in any equilibrium analysis. A particularly simple example (not related to turnout) is the '$k$-pragmatist' heuristic~\cite{reijngoud2012voter}, and there are many others, see a recent survey in \cite{2018Meir}. Some of these essentially model the decision as a function of the (estimated) margin between candidates~\cite{bowman2014potential,fairstein2019modeling}.
 
\medskip
Our key takeaway from the long list of existing models, along the entire `rational-to-heuristic' spectrum, is that the estimated \emph{margin} plays a major role, in addition to the size of the population (that is not always considered). There is also empirical evidence of the connection between (narrow) margin and (high) turnout~\cite{Aldrich}, and strong experimental evidence that the decision to abstain is positively correlated with a high margin and with large population~\cite{levine2007paradox}. \citet{gerber2020one} conducted a large field experiment, that  showed people substantially over-estimate  the chance of a small margin (and thus their chances of being pivotal). %\footnote{\citet{gerber2020one} find that voters' estimate of the margin has no effect on their actions. This lack of correlation in \emph{survey data} (also noted by \citet{Ferejohn75}) contradicts the empirical and experimental evidence above and challenges the Downsian approach to turnout. Though beyond our scope, see \cite{Aldrich} for discussion and reconciliation attempts.}

It is important to mention a stream of papers that assume voters get a noisy signal of the `truth', but actually have shared interest (sometimes called `epistemic voting'~\cite{coleman1986democracy}). In these models a voter may prefer to vote differently from her signal due to Bayesian reasoning, thereby providing additional reasons for failure of the CJT~\cite{austen1996information}. We discuss this in Sec.~\ref{sec:discussion}. However we follow the more common assumption that voters (if they vote) always follow their signal.

\newsubsec{Our Contribution}

%\paragraph{A general class of pivot-estimation models:} 
Analyzing every model from the literature separately would be tedious and leave us with an isolated set of narrow results. Instead, we  stay within the Downsian framework where voters are rational in the sense of aiming to maximize their expected utility in equilibrium, but allow a wide range of ways in which voters estimate their pivotality. The dashed rectangle in Fig.~\ref{fig:boxes} shows the scope of models in our framework. 

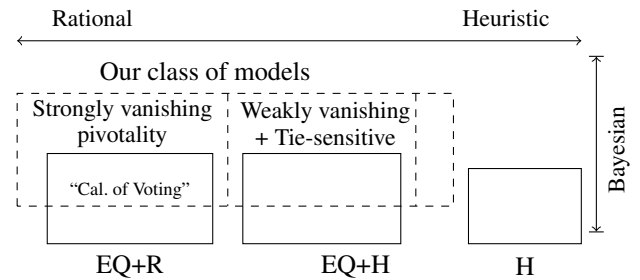
\begin{figure}[ht!]
\centering
\begin{tikzpicture}
\begin{scope}[scale=1]
    % Draw the three rectangles in a row
    \draw (-0.1,0) rectangle (2.1,1.2);
    \node at (01, -0.3) {EQ+R};
    
    \draw (2.5,0) rectangle (4.6,1.2);
    \node at (4, -0.3) {EQ+H};
    
    \draw (5.5,0) rectangle (7,1);
    \node at (6.25, -0.3) {H};

    % Draw the dashed rectangle above
    \draw[dashed] (-0.5,0.5) rectangle (5.3,2);
    \node at (2, 2.3) {Our class of models};
    \node at (0.9,1.6)  {\parbox{3cm}{\centering \small Strongly vanishing pivotality}};
    \node at (3.6,1.6) {{\parbox{3cm}{\centering\small Weakly vanishing \\+ Tie-sensitive}}};
    \node at (1, 0.7) {\scriptsize ``Cal. of Voting"};

    \draw[dashed] (2.3,2) -- (2.3,0.5);
    \draw[dashed] (4.8,2) -- (4.8,0.5);
    \draw[|<->|] (7.2,2.5) -- (7.2,0.15);
    \node[rotate=90] at (7.5,1.25) {\small Bayesian};
    \draw[<->] (-0.5,2.7) -- (7,2.7);
    \node at (0.5,3) {\small Rational};
    \node at (6,3) {\small Heuristic};
\end{scope}
\end{tikzpicture}
\caption{Position of our work within the literature: EQ=Equilibrium; R=Rational; and H=Heuristic.}
\vspace{-2mm}
\label{fig:boxes}
\end{figure}
\paragraph{Our Results:} We fully characterize the CJT result for the models considered in the paper.  At the heart of our characterization result  lie two key elements: (1) the pivot points, which emerge as the limiting points of the sequence of non-trivial equilibria with increasing population size , and (2) the dependence of voter pivotality on two parameters: population size $n$ and margin of victory $m$. While strongly vanishing pivotality models (Calculus of voting model, for instance) tends to collapse the equilibrium into a unique and trivial outcome, it is the more nuanced, weakly vanishing and tie-sensitive models that  gives rise to richer, more complex equilibria. This  leads to outcomes that fundamentally diverge from the classical implications of Condorcet’s Jury Theorem, revealing new dynamics in collective decision-making under uncertainty. In particular, we show that under weakly vanishing models of pivotality estimation it is possible that a significant fraction of population votes but  fails to converge to selecting better alternative even in the limit (hence also for any population size). That is, there is a nonzero probability of a {\em surprise} in the election irrespective of the population size.

\newsec{Model}
\label{sec:model}
%\vspace{-4mm}
 We start with the standard Calculus of Voting model, later expanding it to allow heuristic pivotality estimation.  
%\vspace{-4mm}
\newsubsec{Calculus of Voting} 
We study a two-candidate (referred to as $A$ and $B$)   election with   $N$ voters.  Each  voter  is either a supporter of $A$ (prefers ~$A$, i.e.,  $A \succ_{i} B$)  or a supporter of $B$. 
We adapt the classical two-candidate calculus of voting model  as follows.
The action of each voter is to vote for one of the two candidates $\{A,B\}$, or abstain ($\bot$).  
We set the utility of having $i$'s favorite candidate selected as $1$, and denote  as  $c_i := \max \{ 0,\frac{\texttt{G}_i - \texttt{D}_i}{\texttt{V}_i} \} $ the \emph{effective} cost of voting for voter $i$  (with $\texttt{G}_i$, $\texttt{D}_i$ and $\texttt{V}_i $ as defined in Eq.~\eqref{eq:calculus}).

\if 0
The core supporters (voters with zero effective voting cost) derive more utility from voting for their preferred candidate  than costs incurred  in participating in the voting process.  Voters with non-zero effective costs, on the other hand, vote only if the perceived pivotality of their vote exceeds $c_i$.  For mathematical convenience, we  normalize the effective cost of voting to lie between 0 and 1.  
\fi

We consider the effective cost of voting and the preference of voter~$i$ as  an independent sample from a commonly known  joint distribution $\mathcal{D}$ over $[0,1] \times \{A,B\}$ and is denoted by the tuple $(c_i, T_i)$. We assume that $\calD$ has no atoms, except possibly at $c=0$ (``core voters'').

\paragraph{Rational equilibrium}
Note that the only private information a voter has (other than the type distribution, which is common knowledge), is her own type. Moreover, since voting for the less preferred candidate is strictly dominated, we can assume that the only two actions for voter $i$ are $\top$ (vote for $T_i$) and $\bot$ (abstain). A (pure, ex-ante) strategy profile is therefore a mapping $v$ from the type space $[0,1] \times \{A,B\}$ to an action $\{\top,\bot\}$. 
Now, a voting game is composed of a distribution $\calD$ and population size $N$. Every game, together with a strategy profile $v$, induces a joint distribution on the number of votes for $A$ and $B$.

The expected (normalized) utility gain from actively voting is thus exactly the probability that the voter is pivotal $p^0$, minus the voting cost $c_i$. Since $p^0$  in turn also depends on $v$, we define the explicit function $P^0(v)=P_{N,\calD}^0(v)$ that represents the probability that a single player is pivotal\footnote{For simplicity we consider $p_0$ as the probability that both candidates are \emph{exactly} tied with $\frac{N-1}{2}$ votes each, which is indeed the pivot probability when $N$ is even. When $N$ is odd, things get more complicated, but since for large $N$ the differences are negligible, we maintain our simplifying assumption.} under strategy profile $v$ in the game $(\calD,N)$ (we often omit the game from the definition).  

In a (pure) Bayes-Nash equilibrium, each player picks an action that maximizes her expected utility, given the distribution on other players' actions. This distribution, in turn,  is induced by the type distribution of the other players (conditional on the player's realized type) and the strategy profile. In our case:
\begin{itemize}[leftmargin = *]
    \item The type distribution is an i.i.d. sample from $\calD$, and the voter's realized type reveals no further information;
    \item The utility-maximizing action of each voter is to vote ($\top$) iff Eq.~\eqref{eq:calculus} holds for her type.
\end{itemize}

\begin{definition}[Threshold profile]
A strategy profile $v$ is called a \emph{threshold profile} if there is a threshold $c=c(v)$ such that  
     $\forall c_i\in[0,1], \forall T_i\in\{A,B\}, ~~~v(c_i,T_i)=\top$ iff $c_i\leq c$.    
\end{definition}

%We get the following alternative definition:
\begin{observation}
\label{cor:one}
    A threshold  strategy profile $v$ is a Bayes-Nash equilibrium of the Calculus of Voting game $(\calD,N)$, if and only if 
     $c := c(v) = P^0_{\calD,N}(v).$
    \iffalse 
    \begin{enumerate}[noitemsep]
    \item $v$ is a threshold profile; and
  \item $c := c(v) = P^0_{\calD,N}(v)$.
    \end{enumerate}
    \fi 
    
\end{observation}
\begin{proof}
Let $i$ be a  voter ~$(c_i,T_i)$ and 
note that $P^0(v)$ does not depend on her type $T_i$. As $v$ is a threshold profile,  we have  $p_i= p^0=P^0(v)$.  The utility-maximizing action of voter  is to vote if and only if $c_i\leq P^0(v)$.  
\end{proof}
We  restrict our attention to threshold profiles $c\in [0,1]$ unless explicitly said otherwise. Note that we did not explicitly say what is the function $P^0$. However clearly such a function is well-defined, and even without a formal definition it is clear that fixing $v$, $P_{\calD,N}^0(v)$ decreases very rapidly with $N$. This means that only voters with essentially zero cost will vote, regardless of the type distribution. Indeed, this is the well-known paradox of voting, and it is easy to see that it may lead to an almost-certain win of the minority candidate, if  happens to be supported by more low-cost voters. % A more formal statement will follow later  from  our general results (Prop.~\ref{prop:strong_trivial}). 

 \paragraph{Support functions}
While  $\calD$ contains all necessary information regarding the distribution of costs in the population, we would like to present costs in a more intuitive way. 

For $T\in \{A,B\}$ 
let $s_T:[0,1]\rightarrow [0,1]$  be a continuous, non-decreasing function, where $s_T(c)$ should be read as the fraction of the distribution that prefers $T$  and  has individual cost at most $c$. We call $s_T$ the \emph{support function} of $T$, and note that it does not depend on the identity of voter $i$. 
\begin{restatable}{proposition}{support}\label{prop:support}
    Let voter costs and types are sampled i.i.d. from a  distribution $\calD$. Then any distribution $\calD$ induces a unique pair of support functions $s_A,s_B$ with $s_A(1)+s_B(1)=1$, and vice-versa.
\end{restatable}

\newsubsec{Perceived Pivotality}
The informal statement above regarding the negligible turnout is not only disappointing, but also unrealistic, as in practice a substantial fraction of the population usually votes. We would therefore  relax some of the rationality assumptions. 
These assumptions essentially correspond to the two bullets in the equilibrium characterization in Cor.~\ref{cor:one}: The second suggests that voters act based on their true probability of being pivotal; and the first means that they maximize their expected utility given this probability and Eq.~\eqref{eq:calculus}. % by voting iff their cost lower. 
In what follows, we will maintain the utility maximization assumption  but allow voters much more freedom in estimating their pivot probability $P$. 
\if 0

%Here $c_i \in [0,1]$ is the   cost of voting of agent $i$ and $T_i$ is her type.
%If $T_i = A$ (similarly if $T_i = B$) we call voter $i$ a supporter of $A$ (and supporter of $B$ respectively). %Additionally, if $c_i=0$ we call  voter $i$ a {\em core} supporter of candidate $T_i$. 
Without loss of generality, we assume  that (weakly) more voters support $A$  in expectation. 
We therefore consider $A$ as the ``better'' or ``popular'' candidate, and see the outcome where $A$ is elected as preferable.

\medskip
As discussed above, Condorcet Jury Theorem states that  \emph{if all voters vote}, the probability that $A$ wins goes to 1 as the number of voters grows.
However, this is clearly not always true if $A$ voters are more likely abstain than $B$ voters. This, in turn, may depend both on \emph{individual voters' costs}, and on \emph{voters' perceived pivotality} (which we assume to be the same for all voters). 
\fi
%We next describe how each of these factors is captured in our model.

\paragraph{Pivot functions}
We highlighted earlier that the two most important factors that  determine the probability of various outcomes are (1) the size of voting population, $n$; and (2) the margin, $m$. Our simplifying assumption (following e.g.,~\cite{MyersonWeber}) is that voters only consider $n$ and $m$ as expected values. Thus, an \emph{expectation-based Perceived Pivotality Model} (PPM) is specified by a function $p$, which maps any pair of $n$ and $m$ to $p(n,m)\in(0,1]$, and is continuous, non-increasing in both parameters, and strictly decreasing when strictly below $1$.

PPM  quantifies a subjective ex-ante belief of the individual voter about the importance of her vote.  The equilibrium analysis crucially depends on the PPM model under consideration. In Sec. \ref{sec:PPM}, we  provide several concrete pivotality models, which can either approximate the real pivot probability $P^0$ or reflect beliefs and other factors affecting utility. %\footnote{\label{fn:double_bin} One may wonder if the \emph{actual} pivot probability $P^0$ can also be written as a function of $n$ and $m$, and there is a small caveat as we would also need $N$. If we fix the ratio $n/N$ then we get a true pivot probability that is almost identical to the Binomial PPM we later define, and shares the same properties. See Appendix~\ref{apx:CoV}. 
% The true pivot probability can also be written as a function of $m$ and $n$: it follows the Binomial model $p(n',m)$ from Eq.~\eqref{eq:Binom}, but where the number of active voters $n'$ is a random variable sampled from $Bin(N,n/N)$. Numerically this function is nearly identical to Eq.~\eqref{eq:Binom}, and in particular it has strong vanishing pivotality.\rmr{needs proof}
%} 

Replacing $P^0$ with a general PPM $p(n,m)$ allows us to consider a broad set of voters' behaviors. To see why this is useful, we first observe that the expected margin and the number of voters can be easily derived for any threshold profile $c$.
We  define the two following functions:
\begin{align}
    n(c,N)&:=(s_A(c)+s_B(c))N  \label{eq:n_c}; &\text{(expected  voters)}\\
    m(c)&:=\frac{|s_A(c)-s_B(c)|}{s_A(c)+ s_B(c)}\label{eq:m_c}.&\text{(expected margin)}
\end{align} 

\begin{observation}
    Given support functions (i.e. a type distribution) $(s_A,s_B)$, a threshold profile $c$, and population size $N$, the expected number of active voters is $n(c,N)$ and the expected margin is $m(c)$. 
\end{observation}
\newpar{Election equilibrium}
We can now broaden our class of games. An \emph{Election Game} is a tuple $(s_A,s_B,p,N)$, where $s_A,s_B$ are the support functions of the type distribution $\calD$; $p$ is a PPM; and $N$ is the size of the population. For the special case where $p=P^0_{N,\calD}$, we get the Calculus of Voting game. However, in a general election game, a voter votes according to how much she \emph{perceives herself as pivotal}. We extend the equilibrium definition accordingly. % (see difference from the equilibrium  in Cor.~\ref{cor:one} in  \textbf{bold} font).
\begin{definition}
     A strategy profile $v$ is an \emph{election equilibrium} of election game $(s_A,s_B,p,N)$, if  
     \begin{enumerate} 
         \item $v$ is a threshold profile; and
         \item $c:= c(v) = p(n(c,N),m(c))$.
     \end{enumerate}
\end{definition}
 
\begin{restatable}{proposition}{propTwo}\label{prop:eq_exists}
Every election has at least one equilibrium.
\end{restatable}

The proof follows from the fact that $g(c):=p(n(c,N),m(c))$ is a continuous function from $[0,1]$ onto itself and therefore must have a fixed point.

%Denote by $C(I,N)\subseteq [0,1]$ the set of all equilibrium points of election $E=(I,N)$.

\newsubsec{Issues and Elections}
We want to be able to analyze elections as the population size grows. 
An \emph{issue} is a triple $I=(s_A,s_B,p)$. Thus an issue together with a specific population size $N$ defines an election game $(s_A,s_B,p,N)$ (or just $(I,N)$) as above.
Alternatively, an issue can be thought of as a series of election games, one for every population size $N$. Denote by $C(I,N)\subseteq [0,1]$ the set of all equilibrium points of election $E=(I,N)$. 

\begin{definition}[Issue equilibrium]An equilibrium of issue $I$ is a series of points $\overline{c}=(c_N)_{N}$ s.t. $\forall N, c_N\in C(I,N)$, and $\overline{c}$ has a limit. We denote the limit by $c^*$.
\end{definition}

For an issue equilibrium $\ol c$ with limit $c^*$, if  $c_N>c^*$ for all $N$ we say that $\overline{c}$ is a \emph{right equilibrium}. Similarly, if  $c_N<c^*$ for all $N$ we say that $\overline{c}$ is a \emph{left equilibrium}. 
\paragraph{Trivial equilibria}
An equilibrium is \emph{trivial} if its limit is $0$, meaning only core supporters vote. %Otherwise it is nontrivial. %Also, we refer to an equilibrium where $c^*=1$ as a full equilibrium. In full equilibrium, all voters participate in the voting process. 
\newsec{ Perceived Pivotality Models}
\label{sec:PPM}

%In contrast, our proposed PPM is designed to capture two most important  factors determining  voter turnout (see for instance \todo{add references}); size of the election and winning margin. As stated previously,  many empirical studies show that the voter turnout is relatively more in closely contested elections and that the margin  of victory plays crucial part in drawing voters towards the voting  ballots. In our abstention model we capture this observation by an inverse relation between perceived pivotality and expected margin of victory \todo{references}. %That is, voters believe that each vote is increasingly more important in closely contested  elections. %Furthermore, this inverse relation between voter turnout and expected margin of victory holds even in large elections.
%Second, many rational  (including fully rational model) models of voting behaviour (see )predict that increasing the number of votes on ballot implies the decrease in the importance of each individual vote.    

%We first describe models where the perceived pivotality $p(n,m)$  represents the actual tie probability $V_A=V_B$ (so a single vote would determine the  winner of the election). %However this probability is not, strictly speaking, a function of $n$ and $m$.
We first consider models that closely approximate the actual probability that a single voter is pivotal, i.e., the probability of a tie $V_A=V_B$. 

  % In the  proposed PPM  we consider that the perceived pivotality is heuristically determined and  does not always reflect the probability that a single vote would determine the outcome of an election (herein lies the difference from the calculus of voting model). %We next describe support functions and how  $p$ is determined. %In the next subsection we extend the model to allow a heuristic estimation that is sensitive to small margins.
\newsubsec{Fully Rational models}
\label{ssec:fully rational}
We first argue that our model captures the Calculus of Voting as a special case, i.e. that $P^0$ is also a PPM.
\def\ns{\!\!\!\!}
\begin{proposition}
    For every $N$, there is a PPM $p^{CoV}_N$ s.t. for every threshold profile $c$,
    $P^0_{\calD, N}(c)=p^{CoV}_N(n(c,N),m(c))$.
    More precisely, 
     \begin{equation*}
  \ppm^{CoV}_N(n,m):=   \Exp \limits_{n' \sim Bin(N,\frac{n}{N})}  \big [ \Pr\limits_{x\sim \text{Bin}   ( n', \frac{1+m}{2} )} (x = \floor{n'/2})\big].
  \end{equation*}
\end{proposition}
Note that as $N$ grows, $n'$ is highly concentrated around $n$. We can therefore define an approximate version with a PPM $p$ that does not depend on $N$:
%\rmr{ "Fully Rational models" (define and analyze strongly vanishing PPM);}
 %In this paper we consider  two families of pivot probability models (PPM) that  encompass many  of the pivot probability models studied in literature (sections  \ref{sec:related}).   

% %Strongly vanishing PPM assumes  that the voters vote only if the tie probability is large enough, ignoring the margin of victory. %We begin with some examples of strongly vanishing pivot probability models.  
% Perhaps the most common pivot probability model is a Binomial model, corresponding to a known number of voters $n$ who choose independently whether to vote for $A$ or $B$. This is exactly the model used in Condorcet's Jury Theorem and 

 \begin{example}[Binomial PPM] %   Let $\mathcal{I}$ be a given election instance with $s_A(0), s_B(0) > 0$. \footnote{We assume that at least one of $s_A(0)$ and $s_B(0)$ is strictly positive.} %Otherwise if one of $s_A(0)$ and $s_B(0)$ is zero then the pivot probability is zero and hence $c=0$ is a trivial equilibrium.} Fix $c \in [0,1]$ and  let $n:=n(c)$ and $m:= m(c)$. The Binomial PPM model considers the perceived pivotality as the probability of a tied election. That is, 

 \begin{equation} 
  \ppm^{Bin}(n,m):=\Pr_{x\sim \text{Bin} \big ( n, (1+m)/2 \big)}(x = \floor{n/2}).\label{eq:Binom}
  \end{equation}
 \end{example}
  A later model by ~\cite{Myerson1998} suggested drawing the scores of each candidate independently from a Poisson distribution (see Appendix~\ref{apx:CoV}). % with means $N \cdot s_A(c)$ and $N \cdot s_B(c)$ respectively. %Let $k_A$ voters be drawn independently from the Poisson distribution with mean $N\cdot  s_A(c)$ and $k_B$ voters from the Poisson distribution with mean $N \cdot  s_B(c)$. Candidate $A$ wins an election if $k_A \geq k_B$  and $B$ wins otherwise. 
%The perceived pivotality of the individual vote in this can be written in terms of the  probability that the number of supporters of  both candidates is equal. 
 % \begin{example}
 %  [Poisson PPM]% Let $\mathcal{I}$ be a given election instance.   The Poisson PPM  considers the perceived pivotality as the probability that an equal number of supporters are drawn from Poisson distributions with parameters $N \cdot s_A $ and $N \cdot s_B$.
 %  \begin{equation}
 %      \ppm^{Poi}(n,m):=\Pr_{\substack{X_A\sim \text{Poisson}((1+m)n/2)\\ X_B\sim \text{Poisson} ((1-m)n/2)}}(X_A=X_B)\label{eq:Poisson}
 %  \end{equation}
% \end{example}
Conceptually, the Poisson model is more appropriate in situations where voters can abstain (as the total number of active voters is not fixed),  
However it behaves very similarly to the Binomial model, and for our purpose they are almost the same. In fact, all three models belong in a much larger class of PPMs, characterized by \emph{strong vanishing pivotality}:
\begin{definition}[Vanishing Pivotality]\label{def:vp}
 We say that a PPM $p$ has [strong]  vanishing pivotality if   $\lim_{n\rightarrow \infty}p (n,m)=0$ for all $m> 0$  [$ m \geq 0$].
\end{definition} 

As we will later see, issues with  vanishing pivotality (v.p.) always admit a trivial equilibrium. 
Clearly at the trivial equilibrium, Jury theorems are irrelevant: the candidate with more core support always wins with probability that approaches $1$ as the population grows, regardless of who is more popular overall. It is not hard to verify (e.g., using Stirling approximation) that in both the Binomial and Poisson PPMs, $p(n,m)=\Theta(\frac{1}{\sqrt{n}})$ for  $m=0$, and decreases exponentially fast in $n$ for any $m>0$.  %See also a visual demonstration in Fig.~\ref{fig:pivot}. Thus both models have strongly vanishing pivotality.

 \newsubsec{Tie-Sensitive models}
% \gan{TODO:elaborate, explain the upper bound $T$}
We saw that even in the rational models (which have strong vanishing pivotality), the case of $m=0$ is different, with a substantially higher probability to be pivotal. A simple and perhaps more cognitively plausible assumption is that voters \emph{consider themselves pivotal} if the margin is small enough, regardless of the number of voters.
\begin{definition}[Tie-sensitive pivotality]\label{def:tie_PPM}We say that a PPM is $q$-\emph{tie-sensitive} if $p(n,0)\geq q$ for all $n$.
\end{definition}
That is, if the expected outcome is a tie, everyone thinks they are pivotal at least to some extent, regardless of the number of active voters. 
By definition, any PPM has either strong v.p. \emph{or} tie-sensitivity. If it is tie-sensitive and has v.p. we say it has \emph{weak} vanishing pivotality. Tie-sensitivity may occur due to various reasons, and we provide and discuss examples in Section~\ref{sec:PPM_class}.
\section{Characterizing Equilibrium Limits}

\label{sec:charEquilibria}

We begin by characterizing the trivial equilibrium. 
\begin{restatable}{proposition}{propFour}\label{prop:weak_trivial}
    Suppose $s_A(0)\neq s_B(0)$. Any issue with  vanishing pivotality admits a trivial equilibrium.
\end{restatable}
%\rmr{to the appendix.}

\begin{restatable}
{proposition}{propThree}\label{prop:strong_trivial}
    Suppose $s_A(0)+s_B(0)>0$. Any  issue with strong vanishing pivotality admits \textbf{only}   trivial equilibrium.
\end{restatable}

%\newsubsec{Pivot points}
%The proofs of Propositions  \ref{prop:weak_trivial} and \ref{prop:strong_trivial}   are given in Appendix~\ref{apx:proofs}. Intuitively, the proofs of both  propositions rely on the fact that we can always consider sufficiently large population size $N$, thereby making $p(n,m)$ small enough. 

Next, we  show  that  the intersection points of the support functions (where the margin is 0) form the limiting points of equilibria.   Recall that by our assumption there is a finite number of such points (but see Appendix~\ref{apx:overlap}).

\begin{definition}[Pivot Points]
For a given pair of support functions $s_A,s_B$, a \emph{pivot point} is any  $c\in(0,1)$ where $s_A(c)=s_B(c)$.%, and $s_A(c+\eps)\neq s_B(c+\eps)$ for all $\eps>0$. 
\end{definition}

For technical reasons we will assume throughout the paper that there is only a finite number of intersection points where $s_A(c)=s_B(c)$, and that all derivatives  of $s_A,s_B$ are bounded in some environment of each such point. We explain the more general case in Appendix~\ref{apx:overlap}. 

%i.e., these are exactly the points where the margin is 0. 
%A pivot point is a cost $c$ such that if  all voters with $c_i\leq c$ vote then both candidates have equal support, and thus the expected margin at a pivot point  is 0. 
\def\dinit{\delta}
\def\uld{{\ul \delta}}
\begin{theorem}\label{thm:nontrivial_eq}
    Let $I$ be an issue with a PPM $p$ having   weakly vanishing pivotality \emph{and} is $q$-tie-sensitive.  Any pivot point $c^*<q$ has a right-  and left-equilibrium with limit $c^*$. For support functions with finite intersection points, 
    the limit of any equilibrium  is either 0 or a  pivot point.
\end{theorem}
\begin{proof}
We start with the existence of the right equilibrium. The proof for the left equilibrium is symmetric. Let $c^*<q$ be some pivot point of $I$, and let $\dinit>0$. We need to show there is some $N_\dinit$ and some $c_\dinit\in(c^*,c^*+\dinit)$ s.t. $c_\dinit\in C(I,N_\dinit)$. 

%The proof is very similar to that of Prop.~\ref{prop:weak_trivial}. 
Since all derivatives are bounded, there is some open interval $(c^*,c^*+t)$ where $s_A,s_B$ differ, and w.l.o.g. $s_A(c)>s_B(c)$ for any $c\in (c^*,c^*+t)$. Let $\uld:=\min\{t,\dinit,q-c^*\}$ and note that by the definition of pivot point, $\eps:=m(c^*+\uld)>0$. Also, $n(c^*+\uld,N)=(s_A(c^*+\uld)+s_B(c^*+\uld))N<N$. Thus by weakly vanishing pivotality:
$$p(n(c^*+\uld,N),m(c^*+\uld))\leq p(N,\eps)\xrightarrow[N\rightarrow \infty]{} 0,$$
so there is  $N_\dinit$ for which $p(n(c^*+\uld,N_\dinit),m(c^*+\uld))<\uld$.
From tie-sensitivity, for any $N$ (including $N_\dinit$):
$$p(n(c^*\!,\!N),m(c^*))=p(n(c^*\!,\!N),0) \geq p(N\!,0)\geq q > c^*\!+\uld.$$
Let $g(x):=p(n(c^*+x,N_\dinit),m(c^*+x))-(c^*+x)$. Then  $f$ is continuous in  $x\in[0,\uld)$ with $g(0)>0$ and $g(\uld)<0$. From intermediate value theorem there is some $x^*$ where $g(x^*)=0$ and thus $c_\dinit:=c^*+x^*$ is an equilibrium of $(I,N_\dinit)$.
 
In the other direction, assume towards a contradiction that there is a nontrivial equilibrium with limit $\hat c$ that is not a pivot point. Note that by our assumption of finite intersection points, and due to bounded derivatives, $s_A(c)-s_B(c)>\eps>0$ in some interval $[\hat c-\dinit,\hat c+\dinit]$.
Thus in any point in this interval the pivotality goes to 0 for sufficiently large $N$, and in particular is lower than $\hat c-\dinit$, which means it is not an equilibrium of $(I,\ol{N})$ for any $\ol{N}\geq N$. Note that if $q$ is tight (i.e., the PPM is not $q'$-tie-sensitive for any $q'>q$),  points above $q$ cannot be the limit of any equilibrium, as the pivotality at any $c$ is at most $q$.
\end{proof}

So we have a rather complete characterization of equilibria, or at least of their limit points, in every issue. Two natural questions are: (a) whether these equilibria are inherently stable; and (b) are these equilibria ``good'' in the sense of the Condorcet Jury Theorem.

 \subsection{Stability of  Equilibrium Points}
 \label{sec:stability}
% A natural question about (election) equilibria is whether they are \emph{stable}.
 Intuitively, stability means that a small perturbation will not cause us to drift  from the equilibrium point, but to gravitate back to it  \cite{granovetter1978threshold,palfrey1990testing}. 
 
\begin{definition}[Stability (informal)]\label{def:stable_thresh}An equilibrium $c_N$ is \emph{stable}, if there is some $\eps>0$ such that
for any threshold profile  $c$ with $|c-c_N|<\eps$, the trend\footnote{That is, when starting from $c_N$, the best response of voters in $\varepsilon$-neighborhood of $c_N$ gets closer to $c_N$ until convergence as we increase $N$.} at $c$ is towards $c_N$.
\end{definition}
\begin{restatable}{theorem}{stableEq}
\label{thm:stable}
         Let $\ol c$ be a  right-equilibrium of issue $I$ with limit $c^*>0$. Then for sufficiently large $N$, $c_N$ is stable. 
 \end{restatable}
We provide here the proof outline. The full proof with the exact definition of stability  is given  in Appendix~\ref{apx:stable}. Note that close to the equilibrium point, both $n(c)$ and $m(c)$ are increasing in $c$, %(in fact for linear support functions it is easy to see this is true everywhere above $c^*$), 
and thus $p(n(c),m(c))$ is \emph{decreasing} in $c$. So any perturbation that ends up with fewer active voters will mean higher pivotality, and some voters will join back (and vice versa). 
% Interestingly, a similar proof will \emph{not} work for the left-equilibrium of $c^*$, since $m(c)$ and $n(c)$ have opposite directions. In fact we conjecture that the left equilibrium is always unstable for sufficiently large $N$. 
% A similar proof will \emph{not} work for the left-equilibrium of $c^*$, since as we decrease $c$, we get that $n(c)$ decreases but $m(c)$ grows. Thus $p$ may not be monotone in $c$ and stability may depend on the exact shape of the support functions and the PPM. 
 Since for sufficiently large $N$ the left- and right-equilibria are the \emph{only} equilibria, and the right ones are stable, the left ones must be unstable.

\newsec{Jury Theorems for Pivot Points}\label{sec:CJT}
%As observed above, a trivial equilibrium admits a jury theorem if and only if $A$ has more core supporters, but this is not very interesting. 
Given an election instance $E=(I,N)$ and a threshold profile $c\in [0,1]$, a random variable  counting the number of active votes for $A$ (and likewise for $B$)
 $V_A := \sum_{i\in N}\mathbb{1}[c_i\leq c \wedge T_i=A]$.
%Here $\{(c_i,T_i)\}_{i\in N}$ are voters sampled from $\calD$.
 %Note that $V_A,V_B$ are jointly sampled from a Multinomial distribution. See Sec.~\ref{sec:CJT}. 
 We further denote the \emph{Winning Probability} of $A$ in profile $c$ of election $(I,N)$ as
$\mathbb{WP}_A(I,N,c):=\Pr(V_A > V_B| I,N,c)$. Finally, for any issue $I$ with issue equilibrium $\ol c$, we define
$$\mathbb{WP}_A(I,\ol c):=\lim_{N\rightarrow \infty}\mathbb{WP}_A(I,N,c_N).$$

 We emphasize that we, as `outsiders' to the election, care about the \emph{actual} probability of the event, which is not affected by the perceived pivotality model $p$, once $c$ is determined.

Our main question regards the non-trivial equilibria, whose limits are the pivot points. %Clearly any equilibrium from a side where $B$ has more support admits a non-jury (as $A$ gets less than half the votes, in expectation), but we could still hope that, say, a right equilibrium at a pivot point where the derivative of $s_A$ is higher, would admit a jury theorem. 
%By the results of the previous section, we know that every pivot point has a right equilibrium and a left equilibrium, and that (except the trivial equilibrium) there are no others.  
%  Given a specific pivot point $c^*$, we denote by $m^*:=m'(c^*)$ the derivative of the margin at $c^*$.
% \begin{restatable}{lemma}{lemMbound}\label{lemma:margin_bound}$m(c_N)=|c_N-c^*|(m^*+o(1))$.
% \end{restatable}
% It is also interesting that the exact rate of convergence depends on the derivative of $m(c)$, but not on that of $n(c)$, meaning that the individual slopes of $s_A,s_B$ at $c^*$ don't matter, only the difference between them.
For an equilibrium $\ol c$ with limit $c^*$, we define the convergence rate as $cr(\ol c):=\lim_{N\rightarrow \infty}\sqrt N|c_N-c^*|$.

\begin{definition}
    We say that $\ol c$ is converging \emph{fast} if $cr(\ol c)=0$; and \emph{slow} if $cr(\ol c)=+\infty$.
        Otherwise, we say that $\ol c$ has \emph{moderate convergence rate}
        $r(c^*) \in (0, \infty ) $.
\end{definition}

\def\vph{\vphantom{2^{2^2}_{2_2}}}

\begin{table*}[t]
\setlength{\tabcolsep}{4pt}  % Reduce horizontal padding
\renewcommand{\arraystretch}{1.1}  % Reduce vertical padding
\centering
\begin{tabular}{|p{1.8 cm}|p{1.5cm}|p{1.6cm}|c|c|p{7cm}|}
\hline
\multicolumn{2}{|c|}{\textbf{PPM Type}} & \textbf{Equilibria} & \textbf{Conv. Rate} & \textbf{Jury Thm} & \textbf{Example PPMs} \\
\hline\hline
\multirow{4}{*}{\parbox{2cm}{$q$-tie-sensitive}} 
& no v.p. & at $q$ & - & - & \shortstack{Network (fixed $\kappa$);  Altruist ($f(n)=\exp(\omega(n))$)} \\
\cline{2-6}
& \multirow{3}{*}{\shortstack{weak\\vanishing\\pivotality}} 
& \multirow{3}{*}{\shortstack{at all \\ pivot points \\ below $q$}} 
& slow & Yes & \shortstack{Poly. ($2\beta>\alpha$); Altruist ($f(n)=\omega(\sqrt{n})$)} \\
\cline{4-6}
& & & moderate & Weak & \shortstack{Poly. ($2\beta=\alpha$);  Altruist ($f(n)=\Theta(\sqrt{n})$)} \\
\cline{4-6}
& & & fast & No & Poly. ($2\beta<\alpha$) \\
\cline{1-2} \cline{3-6}
- & strong v.p. & at $0$ & - & - & \shortstack{CoV; Binomial; Poisson;  Altruist ($f(n)=o(\sqrt{n})$)} \\
\hline
\end{tabular}
\caption{(From L to R) First two columns show the three main types of PPM, as per Def.~\ref{def:vp} and \ref{def:tie_PPM}. The third column summarizes the main results of Sec.~\ref{sec:charEquilibria}. Next two columns shows the finer partition of weakly v.p. models, and the winning probability of the leader following the results in Sec.~\ref{sec:CJT}. The rightmost column shows the classification of the models in Sections~\ref{sec:PPM} and \ref{sec:PPM_class}. %, according to the results there and in Sec.~\ref{sec:PPM_class}.
}
\label{tab:results}
\end{table*}

\paragraph{Local Jury theorems}
We say that $I$ admits a \emph{jury theorem} at $c^*\in [0,1]$, if there is  an issue equilibrium $\ol c$ with limit $c^*$ s.t. $\mathbb{WP}_A(I,\ol c)=1$. Similarly, $I$ admits a \emph{non-jury theorem} at $c^*$ if $\mathbb{WP}_A(I,\ol c)<1$, meaning that regardless of the size of the population, there is some constant probability that the better candidate $A$ will lose. $I$ admits a \emph{strong non-jury theorem} if $\mathbb{WP}_A(I,\ol c)=\frac12$. % i.e. if  candidates are equally likely to win.

 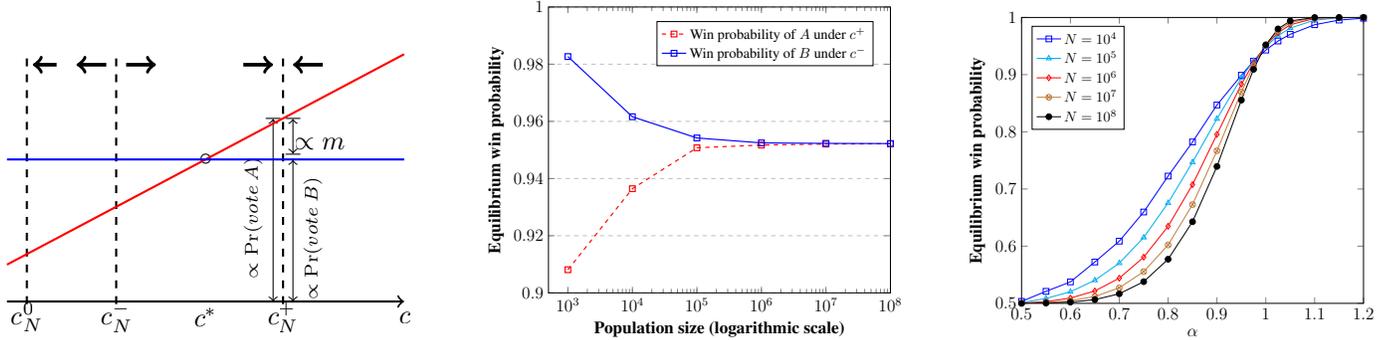
\begin{figure*}[ht]
\centering
\begin{subfigure}{0.28\textwidth}
     \centering 
 \begin{tikzpicture}[xscale=0.66, yscale=0.7, blend group=soft light]
\def\a{40}
\def\bx{-20}
\def\by{-13}
%\draw[pattern=south west lines, pattern color=red] plot[smooth,samples=100,domain=0:5.55] (\x,{3}) -- 
 %   plot[smooth,samples=50,domain=5.55:0] (\x,{0.3});
%\draw[pattern=north west lines, pattern color=blue] plot[smooth, samples=100,domain=0:5.55] (\x,{ 0.97 + 0.5*\x}) -- plot[smooth,samples=100,domain=5.55:5.55] (\x,{0.25})  -- plot[smooth,samples=100,domain=5.55:0] (\x,{0.28})  ;
\draw[blue,thick] (0.5*\a+\bx, 0.4*\a+\by) -- (0.7*\a+\bx, 0.4*\a+\by);
\draw[red,thick] (0.5*\a+\bx, 0.35*\a+\by) -- (0.7*\a+\bx, 0.45*\a+\by);
\draw[thick,->] (0,0.3) -- (8,0.3);
\draw[dashed,thick] (0.639*\a+\bx,0.25) -- (0.639*\a+\bx,5);
\draw[dashed,thick] (0.555*\a+\bx,0.25) -- (0.555*\a+\bx,5);
%\draw[dashed,thick] (0.6*\a+\bx,0.25) -- (0.6*\a+\bx,5);
\draw[dashed,thick] (0.51*\a+\bx,0.25) -- (0.51*\a+\bx,5);
\draw[double,->] (0.55*\a+\bx,4.8) -- (0.535*\a+\bx,4.8);
\node at (0.55*\a+\bx + 0.6,4.8 +0.4 ) {$ $};
\node at (0.55*\a+\bx - 0.31,4.8 +0.4 ) {$ $};
\draw[double,->] (0.56*\a+\bx,4.8) -- (0.575*\a+\bx,4.8);
\draw[double,->] (0.619*\a+\bx,4.8) -- (0.634*\a+\bx,4.8);
\draw[double,->] (0.659*\a+\bx,4.8) -- (0.644*\a+\bx,4.8);
\draw[double,->] (0.525*\a+\bx,4.8) -- (0.513*\a+\bx,4.8);

%\draw[dotted,thick] (0.639*\a+\bx,  0.418*\a+\by) -- (0, 0.418*\a+\by);
\node at (0.6*\a+\bx, 0.4*\a+\by) {$\circ$};

%\node at (0.628*\a+\bx,-0) {$0.628$};
\node at (0.639*\a+\bx,-0) {$c^+_N$};
%\node at (0.569*\a+\bx,-0) {$0.569$};
\node at (8,0) {$c$};
%\draw[draw=gray!50!white,fill=gray!50!white] 
%    plot[smooth,samples=100,domain=0:1] (\x,{0}) -- 
 %   plot[smooth,samples=100,domain=1:0] (\x,{1});
\node at (0.555*\a+\bx,-0) {$c^-_N$};
\node at (0.51*\a+\bx,-0) {$c^0_N$};

%\node at (0.6*\a+\bx,-0) {$0.6$};
\node at (0.6*\a+\bx,-0) {$c^*$};
%\node at (-0.4, 0.4*\a+\by) {$0.4$};
%\node at (-0.5, 0.35*\a+\by) {$0.35$};
%\node at (-0.6, 0.414*\a+\by) {$0.414$};
%\node[blue] at (0.7*\a+\bx, 0.4*\a+\by +0.3) {$s_B = 0.4$};
%\node[red] at (0.7*\a+\bx, 0.45*\a+\by+0.3) {$s_A = 0.1 + \frac{c}{2}$};
\draw[|<->|] (0.639*\a+\bx+0.2, 0.4*\a+\by+0.09) -- (0.639*\a+\bx+0.2, 0.414*\a+\by + 0.23 );
\draw[<->|] (0.639*\a+\bx+0.2, 0.3) -- (0.639*\a+\bx+0.2, 0.4*\a+\by + 0.01 );
\draw[<->|] (0.639*\a+\bx-0.2, 0.3) -- (0.639*\a+\bx-0.2, 0.414*\a+\by + 0.23 );
\begin{scriptsize}
\node[rotate=90] at (0.639*\a+\bx+0.6,1.5) {$ \propto \Pr({vote}~B)$};
\node[rotate=90] at (0.639*\a+\bx-0.6,1.9) {$ \propto \Pr({vote}~A)$};
\end{scriptsize}
\node at (0.639*\a+\bx+0.75, 0.4*\a+\by+0.3) {$\propto m$};
% \tikzstyle{dot}=[rectangle,draw=black,fill=white,inner sep=0pt,minimum size=4mm]
% \tikzstyle{active}=[circle,draw=blue,fill=blue,inner sep=0pt,minimum size=2mm]
% \tikzstyle{inactive}=[circle,draw=blue,fill=white,inner sep=0pt,minimum size=2mm]
% \tikzstyle{sybil}=[rectangle,draw=black,fill=red,inner sep=0pt,minimum size=1.6mm]
% \tikzstyle{del}=[triangle,draw=blue,fill=white,inner sep=0pt,minimum size=2mm]
% \tikzstyle{virtual}=[diamond,draw=black,fill=gray,inner sep=0pt,minimum size=2mm]
% \tikzstyle{txt}=[text width = 8cm, anchor=west]

% \node at (2,0.5) {$r$};
% \node at (2,0) {$*$};
% \node at (5,0) {$*$};
% %\node at (5,-0.5) {$\calG(H)$};
% %\node at (4.4,1.35) [txt] {$\overline \calG_\beta(H)$};

% \draw[thick] (3.6,0) -- (5.8,0);
% \draw[dashed] (3.6,0) -- (2,0);

% \draw [decorate,decoration={brace,amplitude=4pt,raise=4pt},yshift=0pt]
% (3.6,0.1) -- (5.8,0.1) node [black,midway,yshift=0.8cm] {
% $\overline \calG_\safe(H)$};

% \draw [decorate,decoration={brace,amplitude=4pt,raise=4pt},yshift=0pt]
% (5.8,-0.1) -- (2,-0.1) node [black,midway,yshift=-0.8cm] {
% $\calB(r;\overline \calG_\safe(H))$};
\end{tikzpicture}
 
%\subcaption{The pivot point $c^*$ (circle at the intersection) with the two non-trivial equilibria for some specific $N$ on its sides (dashed lines), and the trivial equilibrium $c^0$. For $c^+_N$, the probability of a random voter to vote $A$ is proportional to $s_A(c^+_N)$. The $m(c^+_N)$ is   proportional to the margin of victory. The bold arrows  indicate that $c^+_N, c^0_N$ are stable  equilibria whereas $c^-_N$ is not stable.
%}
\label{fig:linearSupport}
\end{subfigure}
\hfill 
\begin{subfigure}{0.28\textwidth}

    \centering
 \begin{tikzpicture}[scale=0.54] 
\begin{axis}[
    xmode=log,
    log ticks with fixed point,
    xlabel={\textbf{Population size (logarithmic scale)}},
    ylabel={\textbf{Equilibrium  win probability}},
    xmin=500, xmax=10^8,
    ymin=0.9, ymax=1,
    label style={font=\Large},
    tick label style={font=\Large}, 
    xtick={10^3, 10^4, 10^5, 10^6,10^7,10^8},
    xticklabels = {$10^3$,$10^4$, $10^5$, $10^{6}$,$10^7$,$10^8$},
    ytick={0.88,0.9,0.92,0.94,0.96,0.98,1},
    legend pos=north west,
    ymajorgrids=true,
    grid style=dashed,
    legend style={
        at={(axis description cs:0.95,0.95)},
        anchor=north east
    }
]
\addplot[thick, 
    color=red,
    mark=square, 
    mark options={solid},  
    dashed
]
    coordinates {
        (10^3, 0.90810929)(10^4, 0.93644931)(10^5, 0.95076779)(10^6, 0.95171826)(10^7, 0.95206293)(10^8, 0.95216058)
    };
\addplot[thick,
    color=blue,
    mark options={solid},
    mark=square
]
    coordinates {
        (10^3, 0.98266767)(10^4, 0.96161482)(10^5, 0.95419222)(10^6, 0.95247038)(10^7, 0.95230315)(10^8, 0.95223574)
    };
\legend{Win probability of $A$ under $c^{+}$, Win probability of $B$ under $c^-$}
\end{axis}
\end{tikzpicture}

%    \caption{Win probability for different values of $N$ under respective induced equilibria.}
    \label{fig:win-probability}
\end{subfigure}
\hfill 
\begin{subfigure}{0.28\textwidth}

 \begin{tikzpicture}[scale=0.54]
    % Because we need to give the spy node a name to add the labels afterwards,
    % it is a bit more complicate than usual. To do so we need to `scope` the
    % spy. To avoid further error messages it seems we need to `scope` the whole
    % `axis` environment.
    \begin{scope}[
        % Give the spy options to the `scope`
        %spy using outlines={
         %   rectangle,
          %  magnification=10,
           % connect spies,
           % size=2cm,
           % blue,
        %},
    ]
        \begin{axis}[
            ylabel={\textbf{Equilibrium win probability}},
            xlabel={$\alpha$},
            xmin=0.5, xmax=1.2,
             ymin=0.5, ymax=1,
             label style={font=\Large},
                    tick label style={font=\Large}, 
            % (simplified this statement)
            xtick={0.5,  0.6,  0.7, 0.8,  0.9,   1,  1.1,  1.2}, legend pos=north west
            % (removed all unnecessary/unrelated stuff)
        ]

%\begin{axis}[
%    title={Temperature dependence of CuSO\(_4\cdot\)5H\(_2\)O solubility},
 %   xlabel={$\alpha$},
  %  ylabel={Equilibrium win probability},
  %  xmin=0.5, xmax=1.5,
  %  ymin=0.5, ymax=1,
  %  xtick={0.5,0.6,0.7,0.8,0.9,1, 1.1, 1.2, 1.3, 1.4},
  %  ytick={0,0.5,0.6,0.7,0.8,0.9},
  %  legend pos=north west,
  %  ymajorgrids=true,
  %  grid style=dashed,
%]
%fist plot
\addplot[ 
    color=blue,
    mark=square
    ]
            % (simplified the plot by removing a lot of coordinates and adding
            %  `smooth` to the options
%    x = [0.5, 0.625, 0.75, 0.875, 1, 1.25, 1.5, 1.75]
%y1 = [0.5, 0.541341, 0.58145,0.66935,0.95108, 1,1,1]
    coordinates {  (0.5,0.5034463993584801)(0.55, 0.5209968252423642)(0.6, 0.5374775010447751) (0.65, 0.572291459126275)( 0.7, 0.6087449472248506 ) (0.75, 0.6597175923323482)(0.8,0.7225977102031234) (0.85, 0.7822348675774462)(0.9,0.8469174102819855) (0.95,0.8989513770811633) (0.975,0.9235848422584841)(1,0.9429125196891085) (1.025, 0.9587758575127688) (1.05,0.9705707189361287)(1.1,0.9873102229031069) (1.15, 0.9954117338584801 ) (1.2,0.9986352165177431)
    };
%    [0.5034463993584801, 0.5209968252423642, 0.5374775010447751, 0.572291459126275, 0.6087449472248506, 0.6597175923323482, 0.7225977102031234, 0.7822348675774462, , , , , , , 

%second
    \addplot[
    color=cyan,
    mark=triangle
    ]
    coordinates {  (0.5, 0.5010963154212058) (0.55,0.5085490588469159) (0.6,0.5202529716971387)  (0.65, 0.5403472945607244) (0.7, 0.5702710110874876) (0.75, 0.6151564360550874)  (0.8, 0.6756813110811366) (0.85, 0.7469505587253285) (0.9,  0.8227145240083575)  (0.95, 0.8941930364247158) (0.975, 0.9241895067950894)  (1, 0.9492574754059184)  (1.025, 0.9681185863906605) (1.05, 0.9815032301789244)  (1.1, 0.9953040321681711)(1.15, 0.999242990990767)(1.2, 0.9999299920173971)
%[
    };
%third
    \addplot[
    color=red,
    mark=diamond
    ]
    coordinates { (0.5,0.5003468902059904) (0.55, 0.5029925155005879) (0.6, 0.509174210623359) (0.65, 0.521629510461562) (0.7, 0.5439623696839793)(0.75, 0.5807451241705572)  (0.8, 0.6348391717851638) ( 0.85, 0.7076024419027669) (0.9,  0.7954745493948241) (0.95, 0.8832932944912986) (0.975, 0.9207629392877437) (1,  0.9511620674243744)(1.025, 0.9732500282846483) (1.05, 0.9871896441097167)(1.1,  0.9982903515939576) (1.15,  0.9999105144690333) (1.2, 0.9999986813739195)
    
    %[ 
    };
%fourth
    \addplot[ 
    color=brown,
    mark=otimes
    ]
    coordinates {  (0.5, 0.5001097028971535) (0.55, 0.5012324305739424) (0.6, 0.5045321111698506)( 0.65,  0.5119818568953474)( 0.7, 0.5271480860811699)( 0.75,  0.5554127150201131)( 0.8,  0.6021369500337498)(0.85,  0.6728180753803086)(0.9, 0.7668542167736896)(0.95, 0.8696779614724737)(0.975,  0.9154547179317182)(1, 0.9518786838592211) (1.025,  0.9769308242369477)(1.05,  0.9910949188400563)(1.1, 0.9994568515423824)(1.15, 0.9999940908764602)(1.2, 0.9999999949959482)
    %[]
    };
    %fifth 
    \addplot[ 
    color=black,
    mark=*
    ]
    coordinates {  (0.5, 0.5000346912422677)(0.55, 0.5005458562520027)(0.6,  0.5020799539199186)(0.65, 0.506433665179634)(0.7, 0.5167118739393539)(0.75, 0.5378743422746576)(0.8,  0.5771469031991653)(0.85, 0.6427745084476345)(0.9,  0.7394354680494846)(0.95, 0.8554562683942735)(0.975, 0.9093048027818923)(1, 0.9521161311947801)(1.025, 0.979997219313373)(1.05, 0.9939073988915024)(1.1,  0.9998608347908574) (1.15, 0.999999832098464)(1.2, 0.9999999999983229)
    %[]
    };

            % crate a coordinate of the point you want to magnify
            %\coordinate (point) at (axis cs:0.755,0.712);
     %       \PointX = 1.25
      %      \PointY = 0.99963
%            \Getxycoords{point}{\PointX}{\PointY}
            % Get the coordinates of that point (to later use them)
%            \Getxycoords{point}{\PointX}{\PointY}

            % draw the dashed lines to the axis (using the defined coordinate)
%            \draw [red,dashed]
 %               (point -| {axis cs:\pgfkeysvalueof{/pgfplots/xmin},0})
  %                  -| ({axis cs:0,\pgfkeysvalueof{/pgfplots/ymin}} -| point);

            % unfortunately one cannot directly place the spy at an
            % axis coordinate, thus we define a `\coordinate` first
                %\coordinate (spy point) at (axis cs:0.76,0.53);
                %\Getxycoords{point}{\PointX}{\PointY}
            %\spy[color=red] on (point) in node (spy) at (spy point);
            \addlegendentry{$N = 10^4$}
\addlegendentry{$N = 10^5$}
\addlegendentry{$N = 10^6$}
\addlegendentry{$N = 10^7$}
\addlegendentry{$N = 10^8$}
        \end{axis}
    \end{scope}
\end{tikzpicture}

     %\caption{Win probability of  $A$ for $\beta = 0.5$ and different values of $\alpha$ in polynomial PPM model  for different values of  $N$. The trend reversal,  where win probability of $A$ for largest $N$ shifts from lowest to highest,  can be observed at   $\alpha =1$.   }
     \label{fig:fourth}
\end{subfigure}
\caption{(L)  demonstrates election instance from Section~\ref{sec:simulation}. For  large value $N$, the pivot point $c^*$, two non-trivial equilibria $c_N^+$ and $c_N^-$, and the trivial equilibrium $c^0$ are shown. %For $c^+_N$, the probability of a random voter to vote $A$ is proportional to $s_A(c^+_N)$. The $m(c^+_N)$ is   proportional to the margin of victory.
The bold arrows  indicate that $c^+_N, c^0_N$ are stable  equilibria whereas $c^-_N$ is not stable. (C) Win probability for different values of $N$ under respective induced equilibria. (R) Win probability of  $A$ for $\beta = 0.5$ and different values of $\alpha$ in polynomial PPM model  for different values of  $N$. The trend reversal can be observed at   $\alpha =1$.  \label{fig:sim} 
}
\end{figure*}

\begin{restatable}[Characterization of Stable Jury Equilibria]{theorem}{JuryCharacterization}\label{thm:JT}
    Let $I$ be an issue and let $\ol c$ be a nontrivial right equilibrium of $I$ with limit $c^*$. 
    There are three cases, where $\ol c$ admits a 
    
    \begin{enumerate}
        \item  \emph{Jury theorem} if $\ol c$ converges slowly;
        \item \emph{weak non-Jury} theorem if $\ol c$ converges moderately; and
         \item \emph{strong non-Jury theorem} if $\ol c$ converges fast.
    \end{enumerate}
\end{restatable}
\begin{proof}[Proof Sketch]
\iffalse 
Let $\mu$ and $\sigma$ denote the mean and standard deviation of $X:=V_A-V_B$, respectively, approximated  by a Normal random variable. We use the ratio $\frac{\mu}{\sigma}$ to determine whether $\overline{c}$ admits a particular variant of the Jury theorem. In particular, when $m(c_N) $  converges slowly i.e. with rate $\omega(1/\sqrt{N})$ then we show that $ \Phi(\frac{\mu}{\sigma}) \rightarrow 1$ as $N $ increase to infinity. Here, $\Phi(.)$ is a CDF of the Normal random variable. This implies that $\overline{c}$ admits Jury theorem as in this case $\mu/\sigma$ goes to infinity as $N$ grows.

%Under fast convergence i.e. when $m(c_N) = o(1/\sqrt{N})$ we have that $\mu/\sigma$ goes to zero as $N$ grows. This gives $\Phi(\frac{\mu}{\sigma}) \xrightarrow[N\rightarrow \infty]{} 0 $, implying weak non-Jury theorem. Finally, under moderate convergence rate i.e., when $m(c_N) = \Theta(1/\sqrt{N})$ then $\frac{\mu}{\sigma}$ has a finite positive limit implying that $\Phi(\frac{\mu}{\sigma}) <1$. 
\fi 
Let $\mu$ and $\sigma$ denote the mean and standard deviation of the random variable $V_A-V_B$ representing the difference between the votes received by two candidates. We approximate this random variable by a Gaussian  distribution. The ratio $\frac{\mu}{\sigma}$ is critical for determining whether the sequence $\overline{c}$ satisfies a specific variant of the Jury theorem, as the winning probability is approaching $\Phi\left(\frac{\mu}{\sigma}\right)$ (here $\Phi(\cdot)$ denotes the CDF of standard normal distribution).

We prove  that the margin $m(c_N)$ is essentially proportional to $|c_N-c^*|$.
In particular, when the sequence $|c_N-c^*|$ converges at a slow rate, so does the margin, meaning that $m(c_N) = \omega(1/\sqrt{N})$, it can be shown that $\frac{\mu}{\sigma}$ tends to infinity with $N$, and consequently $\Phi\left(\frac{\mu}{\sigma}\right)$ tends to $1$. This result implies that $\overline{c}$ adheres to the Jury theorem since, in this case, the ratio $\frac{\mu}{\sigma}$ grows unbounded as $N$ increases.

 On the other hand, under conditions of fast  convergence, where $m(c_N) = o(1/\sqrt{N})$, the ratio $\frac{\mu}{\sigma}$ approaches zero as $N$ becomes large. Consequently, $\Phi\left(\frac{\mu}{\sigma}\right) \rightarrow 0$, indicating weak non-Jury theorem. This result reflects that the distribution of the voting difference becomes more centered around zero with a diminishing spread, leading to an increasingly uncertain outcome.

For the intermediate case of moderate convergence, where $m(c_N) = \Theta(1/\sqrt{N})$, the ratio $\frac{\mu}{\sigma}$ converges to a finite positive limit as $N$ increases. This behavior implies that $\Phi\left(\frac{\mu}{\sigma}\right)$ remains bounded strictly between 0 and 1. Consequently, $\overline{c}$ does not fully satisfy the Jury theorem but neither does it completely diverge from its principles.  A detailed proof is given in Appendix \ref{apx:proofs}.  
\end{proof}

 \newsec{Classification of PPMs}\label{sec:PPM_class}

 Recall the `rational' PPMs with strong vanishing pivotality we considered in Sec.~\ref{sec:PPM}. As our positive results apply for PPMs that are tie-sensitive, it is important to at least provide some examples of such  models, that can also be justified in practice. 
 
  We suggest simple models demonstrating how tie-sensitivity may emerge from at least three reasons: limited communication, altruism, and heuristics. \rmr{A fourth reason is a belief that votes are correlated. }

\paragraph{Limited Communication}
Several authors in the literature considered models in which voters are embedded in an implicit or explicit network, where they only `see' a limited number of $K$ neighbors~\cite{OSBORNE2003434,michelini2022group}, based on which they can assess their pivotality. Generally, $k$ out of $K$ neighbors will be active in expectation, and $k$ can be a function $k=\kappa(n)$. This  yields the following model:%\footnote{While in general $\kappa(n)\neq \frac{n}{N}\kappa(N)$, in the limit of every equilibrium, $n$ is a fixed fraction $\tau$ of $N$, and thus $k=\tau\cdot \kappa(\frac{1}{\tau}n)$ which is also some function of $n$ only.}

\begin{example}[Network PPM]
    $p^{Network(\kappa)}(n,m):=p^{Bin}(\kappa(n),m)$.
\end{example}
It is not hard to see that if e.g. $k=\kappa(n)$ is a constant, then the model is $q_k$-tie-sensitive for some $q_k=\Theta(\frac{1}{\sqrt{k}})$ and does not have vanishing pivotality at all.  Thus equiliria are not at pivot points.

Otherwise (i.e. $k$ is strictly increasing with $n$), it has strong vanishing pivotality and thus only the trivial equilibrium.

\paragraph{Altruist voters} Another possibility is that voters correctly estimate their pivotality (say, using the Binomial or Poisson model above), but that their \emph{value} $\texttt{V}_i$  scales with the size of the population as $\texttt{V}_i \cdot f(N)$. That is, voters consider large elections as `more important'. We note that this argument is sometimes used as a possible explanation for the paradox of voting~\cite{downs57}. 
Note that
$$ \frac{\texttt{G}_i-\texttt{D}_i}{\texttt{V}_i \cdot f(N)} < p \iff  c_i = \frac{\texttt{G}_i-\texttt{D}_i}{\texttt{V}_i}< p \cdot f(N). $$

We therefore get another class of PPMs (considering that $n$ and $N$ scale roughly at the same rate):
\begin{example}%[Altruist PPM]
    $p^{alt(q,f)}(n,m):=\min\{q,  p^{Bin}(n,m)\cdot f(n)\}$.
\end{example}

It is not hard to see that the Altruist PPM is weakly vanishing for any sub-exponential function $f$, and that it is $q$-tie-sensitive whenever $f=\Omega(\sqrt n)$.  In fact, we can classify all the regimes of altruist PPMs:%(and in particular when $f$ grows proportionally to the population size).

% \begin{restatable}{lemma}{altruistic}
% \label{lemma:alt_vanish}
%     The Altruist PPM has v.p. iff $f=\exp(o(n))$  and is $q$-tie-sensitive for some $q>0$ iff $f=\Omega(\sqrt n)$. 
% \end{restatable}

%\rmr{move proof to appendix}

 % \paragraph{Altruist PPM:} The `altruism' functions $f$ for which the PPM has weakly v.p. is when $f$ grows at sub-exponential rate, but at least like $\sqrt N$. 

\rmr{the exact boundary between the two first conditions is more nuanced. the theorem states c is nontrivial but condition 4 says trivial}
\begin{restatable}{proposition}{PropAltruistic}
    Let $I$ be an issue with an Altruist PPM with function $f$, and let $\ol c$ be a non-trivial equilibrium. Then: 
    \begin{enumerate} 
        \item if $f=e^{\omega(n)}$ then $\ol c$ is a fixed equilibrium at $q$; else
        \item if $f=\omega(\sqrt n)$ then $\ol c$ converges slowly; else
        \item if $f=\Theta(\sqrt n)$ then $\ol c$ converges moderately; else
        \item $\ol c$ is trivial.
    \end{enumerate}
\end{restatable}
Interestingly, there is no $f$ for which there is a non-trivial equilibrium with fast convergence (See also Table~\ref{tab:results}). 
\paragraph{Polynomial heuristics}
Another PPM is induced by defining the dependency on $n$ and $m$ directly: 
\begin{example}[Polynomial PPM] For $q,\alpha,\beta>0,  p^{Poly(q,\alpha,\beta)}(n,m)=\min\{q,\ m^{-\alpha}\cdot  n^{-\beta}\}.$ 
\end{example}
It is immediate from the definition that the Polynomial PPM is both $q$-tie-sensitive and has a weakly vanishing pivotality, so it remains to classify the models according to rate of convergence. %\footnote{1-sensitivity may sound extreme but note that it is only a matter of scaling the utility voters get from changing the outcome ($V_i$) vs. the costs and benefits of voting ($D_i-G_i$).}  

Suppose that the support functions have different derivatives at $c^*$ (see Appendix~\ref{apx:overlap}).
\begin{restatable}{lemma}{lemNonTrivial}\label{lemma:poly_margin}Let $I$ be an issue with a Polynomial PPM and let $\ol c=(c_N)_N$ be a  non-trivial equilibrium of $I$ with limit $c^*$. Then $c_N=c^*+\Theta(N^{-\frac{\beta}{\alpha}})$.
\end{restatable}
From the lemma (and Theorem~\ref{thm:JT}),  we  conclude that the critical threshold for the Jury theorem is at $\alpha=2\beta$: If $\alpha>2\beta$ then convergence is fast and there is a strong non-jury theorem; and if $\alpha<2\beta$ then convergence is slow enough to guarantee a local jury theorem. See Table~\ref{tab:results}.

The threshold includes the special case where $\beta=\frac12, \alpha=1$. This  is case is interesting and natural because the dependency on the number of active voters is $\frac{1}{\sqrt n}$---just as in the fully rational models---whereas the dependency on the margin is linear.

Moreover, in the  case of moderate convergence rate where $\alpha=2\beta$,  $A$ wins w.p. $\Phi((c^*)^{-\frac{1}{\alpha}})$. Conveniently, most terms cancel out and we get that $\mathbb{WP}_A(I,\ol c)=\Phi ((c^*)^{-\frac1\alpha})<1$. Interestingly, the probability does not depend on the shape of the support functions at all (except  their intersection point), and neither does it depend on  $\beta$.  

\subsection{Sensitivity to model parameters} \label{sec:simulation}
Our theoretical results provide a sharp threshold between `positive' and 'negative' results, which depends on the PPM or its parameters. However, there results hold \emph{in the limit} as $N$ grows. 

We next study via an example how $A$'s winning probability changes as population size increases and/or when we vary the parameters. For simplicity and easy computation we use the Polynomial PPM with $\beta=\frac12$.  

We set $\calD$ such that each voter supports  $A$ w.p. 0.6  overall, of which $\frac16$ of the voters (10\% of entire population) are core supporters, and the rest $\frac56$ have $c_i\sim U[0,1]$, see Fig.~\ref{fig:sim}(L).

Note that the unique pivot point is $0.6$, so we expect the winning probability  to approach $\Phi(0.6^{-1})\cong 0.952$ in both $c^+_N$ and $c^-_N$ as $N$ increases.

%We start by analyzing the example assuming the Tie-sensitive PPM with $\alpha=1$, as assumed throughout the analysis in the previous sections.
\paragraph{Insights from the empirical example}
While the winning probability of the equilibrium winner indeed approaches its limit value, it is increasing with $N$ for the right equilibrium (where $A$ wins), but \emph{decreases} towards $c^*$ in the left equilibrium $c^-$ (Fig.~\ref{fig:sim}(Center)). 

In the right panel of the figure we observe that with  higher values of $N$ equilibrium win probability of at $c^+$   approaches a step function around $\alpha=1=2\beta$. Yet this occurs slowly and   even for  high $N$ there is  a gradual increase of the win probability of $A$ with $\alpha$.
\iffalse 
\paragraph{Partial turnout equilibrium}
Note that the pivot point is $c^*=0.6$, so we expect the winning probability  to approach $\Phi(0.6^{-1})\cong 0.952$ in both $c^+_N$ and $c^-_N$ as $N$ increases ( 
Fig.~2 (C)).  % demonstrates the Non-Jury Theorem for Example~\ref{example:four}, showing that the winning probability of A in the stable equilibrium $c^+$ is bounded by $\cong 0.955$, even as the  population size $N$ is increasing. 
Interestingly, the winning probability of $B$ behaves differently and  \emph{decreases} towards $c^*$ as the population grows. 
%indicating that the left equilibrium approaches $c^*$ slower (so the margin remains high for longer). %towards the $c^-$ equilibrium point. Hence  $B$ would always prefer a smaller fraction of the   population to vote whereas the popular candidate prefers a large  population to vote.
% we observe that the  win probabilities under $c^+$ and $c^-$ for respective winning candidates are significantly different. While $B$ wins with a high probability under $c^-$, $A$ wins under $c^+$ with a high probability for the same population size.  Furthermore, for  a fixed population size, the win probabilities of these candidates under their favoured equilibria are not the same.

\paragraph{Varying the PPM parameters}
 Fig.~2 (R) shows  win probability of  $A$ under $c^+_N$ for different  values of $\alpha$ and $N$. Higher values of $N$  approach a step function around $\alpha=1$$=2\beta$. Yet this occurs slowly and   even for  high $N$ we observe a gradual increase of the win probability of $A$ with $\alpha$.
 
\fi

\newsec{Discussion}\label{sec:discussion}
The starting point of this paper was the Paradox of Voting, which entails the Condorcet Jury Theorem (CJT) would not hold in a population with heterogeneous voting costs, due to abstention and bias. 
We showed that when voters are sufficiently responsive to the expected margin, the (locally) more popular candidate wins with a probability approaching one, aligning with the classical CJT.

There is of course ample literature on the CJT analyzing theoretical and practical conditions where it may fail (see \cite{mccannon2015condorcet}).
In the introduction, we already mentioned \emph{epistemic voting}, that is supposedly outside the scope of models we consider: e.g., in epistemic voting minority voters are unaware of this and everyone gains from high turnout. However, a recent paper by Michellini et al.~\citep{michelini2022group} suggest a way to re-obtain a CJT when voters are only exposed to few neighbors.  This can be thought of as a special case of heuristics that increases their perceived pivotality, so perhaps our results could be extended to cover `heuristic' epistemic voting as well, with some modifications.

Strategic behavior in the epistemic voting scenario is more problematic though, leading to the `swing voter curse'~\cite{feddersen1996swing,austen1996information}. It is less likely  that considering self as more pivotal would help, as the problem is with the voter's belief about her \emph{competence}.

\rmr{Miller~\cite{miller1986information} suggests an extension of CJT to a political setting, where there is an initial partition between A and B as in our model, but each voter can `mistake' w.p. $p$. This gives rise to a different distribution of votes even under full participation (that depends both on $n/n_A$ and $p$), but again it might be possible to extend our results. }

 Other factors affecting vote decisions like external pressure~\cite{bolle2022voting},  bandwagon effect~\cite{morton2015exit}, and information cascades ~\cite{golub2010naive,acemoglu2011bayesian}, which provide alternative reasons for the failure of group wisdom,  are outside the scope of our paper.

 % we identified two key parameters that influence most models of pivotality: population size and the expected margin of victory. Our analysis reveals that under certain conditions on these parameters, non-trivial equilibria with substantial turnout can arise—offering a potential resolution the  Paradox of Voting. In particular, 

\paragraph{Future directions}
I n light of the above discussion, we believe that relaxing the rigid rationality assumption in favor of more general pivot probability models is an important step in understanding the boundaries of positive and negative results in the social choice literature (in our case, CJT and paradox fo voting, respectively). 

While  this paper focused on the limit case,  studying  the winning probabilities in small-population settings is also important, following our preliminary empirical example. 

Another avenue is to explore heterogeneity not only in participation costs but also in how voters estimate their pivotality. While it is relatively straightforward to show that an equilibrium still exists in this extended model (Appendix~\ref{apx:diverse}), characterizing these equilibria and understanding how this diversity affects the CJT remains an open and intriguing question.

\bibliography{aaai2026}
\onecolumn
\appendix

 \setcounter{secnumdepth}{2} %May be changed to 1 or 2 if section numbers are desired.

\section{Missing Proofs}\label{apx:proofs}
\label{apx:missingProofs}
\support*
\begin{proof} Suppose there are $N$ voters sampled from $\calD$.
    Given a cost threshold $c$,  let the random variable $X_A^{(i)}(c) := \frac{1}{N-1} \sum_{j \neq i}   \mathbb{1}[ c_j \leq c \text{ and } A \succ_{j} B ]$ denote fraction of supporters---other than voter $i$---of candidate $A$ having the voting cost at-most $c$.  Similarly, define  $X_B^{(i)}(c):= \frac{1}{N-1} \sum_{j \neq i}^{N-1} \mathbb{1}[ c_j \leq c \text{ and } B \succ_{j} A ]$. Note that, as $X_{T}^{(i)}(c)$ does not depend on her type $T_i$ and her private cost $c_i$. Since  every agent  is exposed to the same knowledge, we have, $ X_{T}(c) := X_{T}^{(1)}(c) = X_{T}^{(2)}(c) = \cdots = X_{T}^{(N)}(c)$ for all $T \in \{A,B\}$. Finally, let $ s_T(c) := \mathbb{E}[X_{T}(c)] $, and note that it depends neither on $i$ or $N$.

    To see that $s_T(c)$ is continuous, let $s^+, s^-$ the right and left limits of $s_T(c)$ at $c^*\in(0,1)$. If $s^-<s^+$ then $Pr_{\calD}[(c^*,T)] = s^+ - s^- >0$, in contradiction to our assumption that $\calD$ has no atoms in  $(0,1)$.

    To show the direction, we define $Pr_\calD[T_i=A]:=s_A(1)$, then by assumption $Pr_\calD[T_i=B]:=1-s_A(1)=s_B(1)$. For each $T\in \{A,B\}$, the function $s_T(c)/s_T(1)$ is the CDF of $Pr_\calD[c_i| T_i]$.
\end{proof}

\propTwo*
\begin{proof}
    For any $c\in [0,1]$, define $f(c):=p(n(c),m(c))$. Since the support functions are continuous, $n(c)$ and $m(c)$ are continuous (by Eqs.~\eqref{eq:n_c},\eqref{eq:m_c}). By continuity of $p$, the function $f$ is also continuous. Finally, every continuous function from $[0,1]$ to itself has a fixed point, due to intermediate value theorem applied to $f(x)-x$.
\end{proof}
n size).

\begin{restatable}{lemma}{altruistic}
\label{lemma:alt_vanish}
    The Altruist PPM has v.p. iff $f=\exp(o(n))$  and is $q$-tie-sensitive for some $q>0$ iff $f=\Omega(\sqrt n)$. 
\end{restatable}
\begin{proof}
For any $m>0$, we have  
$$p^{Bin}(n,m)\cdot f(n) = \exp(-\Theta(n))f(n),$$
which goes to 0 if $f(n)$ is subexponential, and goes to infinity otherwise. 

Similarly, at $m=0$, we have
$$p^{Bin}(n,m)\cdot f(n) = \Theta(\frac{1}{\sqrt{n}})f(n),$$
which goes to 0 if $f(n)=o(\sqrt n)$ (in which case the model has strong vanishing pivotality); goes to infinity if $f(n)=\omega(n)$; and approaches some positive constant $q'$ if $f(n)=\Theta(\sqrt n)$, in which case we define $q:=\min\{1,q'\}$.
\end{proof}

\propThree*
\begin{proof}
    Assume towards a contradiction that there is a nontrivial equilibrium $\ol c=(c_N)_N$ with limit $c^*>\delta$ for some $\delta>0$. This means  that $c_N>q$ for all sufficiently large $N>N'$. %However, for any $N$, the number of active voters $n(c_N)$ is at least $(s_A(\delta)+s_B(\delta))N= \Omega(N)$.
    
    Now, since $\lim_{n\rightarrow \infty}p(n,0)=0$, there is some $n^*$ such that $p(n,0)<q$ for all $n\geq n^*$. Let $N^*:=\frac{n^*}{s_A(c_{N^*})+s_B(c_{N^*})}$, so that $n(c_{N^*},N^*) = n^*$. Also w.l.o.g. $N^*>N'$. We then have that 
    \begin{align*}
    q<c_{N^*}& = p(n(c_{N^*},N^*),m(c_{N^*})) \\ & \leq p(n(c_{N^*},N^*),0) \\ & = p(n^*,0) <q,  \end{align*}
i.e. a contradiction.
\end{proof}
\propFour*
\begin{proof}
For any $\delta>0$, we should show that there is some $N_\delta$ and $c_{\delta}<\delta$ such that $c_\delta \in C(I,N_\delta)$.

W.l.o.g. $s_B(0)>s_A(0)$, and by the bounded derivatives, the margin $m(c)$ is at least some $\eps>0$ in some neighborhood $c\in [0,q]$.
Let $\uld:=\min\{q,\delta\}$.

 By weakly vanishing pivotality, 
 $$p(n(\uld),m(\uld))\leq p(n(\uld),\eps) = p(s_B(0)N,\eps)\rightarrow 0$$
 as $N$ grows. In particular there is some $N_\delta$ for which $p(s_B(0)N_\delta,\eps)<\uld$. We argue that there must be an equilibrium in the range $[0,\uld]$.

 Indeed, fix $N_\delta$ and let $f(x):=p(n(x),m(x))-x$. We know that $f$ is continuous, that $f(0)\geq 0$ and that $f(\uld) = p(n(\uld),m(\uld))-\uld< \uld-\uld=0$. From the intermediate value theorem, there must be some $c_\delta\in [0,\uld)$ for which $f(c_\delta)=0$, meaning $c_\delta=p(n(c_\delta),m(c_\delta))$. Thus $c_\delta\in C(I,N_\delta)$.
\end{proof}
%\lemMbound*

\begin{lemma}\label{lemma:margin_bound}$m(c_N)=|c_N-c^*|(m^*+o(1))$.
\end{lemma}
\begin{proof}
Since the first and second derivatives of $s_A,s_B$ are bounded, so are those of the functions $n(c,N)$ and $m(c)$ (as per Eq.~\eqref{eq:n_c},\eqref{eq:m_c}).

    We denote $m'(c):=\frac{\partial m(c)}{\partial c}$ and $n'(c):=\frac{1}{N}\frac{\partial n(c,N)}{\partial c}$, so neither function depends on $N$.
    
    Also note that $n'(c)\geq 0$ everywhere and $m^*=m'(c^*)>0, n'(c^*)>0$ by the definition of pivot point. We also denote $s^*:=s_A(c^*)+s_B(c^*)$ and note that $s^*\cdot N = n(c^*)$.

Due to  bounded derivatives,  there must be some interval $[c^*-\delta,c^*+\delta]$ and constants $\ol m,\ul m,\ol n,\ul n>0$ such that $m'(c)\in [\ul m,\ol m]$ and $n'(c)\in[\ul n,\ol n]$ for all $c\in [c^*,c^*+\delta]$. Moreover, $\ul m,\ol m$ can be arbitrarily close to $m^*$ and likewise for $n'$ (the functions are nearly-linear near $c^*$).  We consider only $N$'s large enough so that $c_N$ is in this interval, and so that $p(n(c_N),m(c_N))<1$.

For any $c\in [c^*,c^*+\delta]$, we have that 
\begin{align*}
    m(c&)\in [m(c^*)+|c-c^*|\ul m, m(c^*)+|c-c^*|\ol m] \\
    &= |c-c^*|\cdot [\ul m,\ol m] = |c-c^*|(m^*+o(1)),
\end{align*}
where the last equality is due to the bounded second derivatives of the support functions, and in particular of their difference.
\end{proof}
\JuryCharacterization*
\begin{proof}
    Let $\ol c=(c_N)_N$ be a nontrivial equilibrium with limit at pivot point $c^*<q$. 
%    The random variable $V_A$ is a Binomial variable with parameters $p_A=0.5(1+m(c_N))$ and  $n=n(c_N)$, and $A$ wins if $V_A>n/2$.\footnote{To be completely accurate, we should sample $v_A$ from a Binomial distribution with a number of samples that is also a Binomial variable (see Footnote~\ref{fn:double_bin}), but the difference is negligible.}

Denote $p_A:=s_A(c_N); p_B:=s_B(c_N)$ the probabilities that a single random voter is voting actively for $A$ or $B$, respectively with $p_A > p_B$ and let $s^*:=p_A+p_B$. 

The  variables $V_A,V_B$ are coming from a single multinomial distribution with parameters $N$ and $(p_A,p_B, 1-(p_A+p_B))$. In the limit of $N$, these are two \emph{correlated} Normal variables, with mean $\mu_A = N\cdot p_A, \mu_B = N\cdot p_B$  and variance $\sigma_A^2=N p_A(1-p_A), \sigma_B^2=Np_B(a-p_B)$. Their difference $V_A-V_B$ is a Normal variable with  expectation 
\begin{align*}
 \mu &= \mu_A-\mu_B= N(p_A-p_B)    = N\cdot m(c_N)s^*
\end{align*}
and with variance
\begin{align*}
  \sigma^2 &= \sigma_A^2+\sigma_B^2-2COV(V_A,V_B) \\
&= N(p_A(1-p_A) + p_B(1-p_B)) +2N(p_A p_B)\\
&= N(p_A+p_B - (p_A^2 -2p_A p_B +p_B^2))\\
&= N(p_A+p_B - (p_A - p_B)^2)\\
&=N(p_A+p_B-(m(c_N)(p_A+p_B))^2)\\
&=N((p_A+p_B)(1-m(c_N)^2(p_A+p_B)))\\
&=n(c_N,N)-O(N\cdot m(c_N)^2) = N\cdot s^*-o(N),
% &=N(p_A+p_B+O(1/N)) \\ 
% &=N(p_A+p_B)+O(1) 
\end{align*}
where the last equality is since $m(c_N)$ is diminishing.

From the calculations above, 
\begin{equation}\label{eq:mu_sigma}
  \frac{\mu}{\sigma}=\frac{N\cdot s^* \cdot m(c_N)}{\sqrt{s^*\cdot N- o(N)}}\cong \sqrt{N}\sqrt{s^*}\cdot m(c_N),
\end{equation}
where  $x\cong y$ is a shorthand for $\lim_{N\rightarrow \infty}|x-y|=0$. This is the main fact we need for all three cases.

    \paragraph{Slow convergence}
    Suppose that $\ol c$ converges slowly to $c^*$. Then by Lemma~\ref{lemma:margin_bound}, $m(c_N)=\Theta(1)|c_N-c^*|=\omega(1/\sqrt N)$, and from Eq.~\eqref{eq:mu_sigma},
    \begin{align*}
      \frac\mu\sigma&=\sqrt{N}\sqrt{s^*}\cdot \omega(1/\sqrt N) = \omega(1),
    \end{align*}
    i.e., $\frac{\mu}{\sigma}$ goes to infinity as $N$ grows. As a result, using Normal approximation of $V_A-V_B$, we have
     \begin{align*}
        &\mathbb{WP}_A(I,N,c_N)\cong \Pr_{x\sim \calN(\mu,\sigma^2)}(x>0) = \Pr_{x\sim \calN(0,1)}(x<\frac{\mu}{\sigma}) \Rightarrow \\
       &\mathbb{WP}_A(I,N,c_N)
       \cong  \Phi(\frac{\mu}{\sigma}) \xrightarrow[N\rightarrow \infty]{} 1,
    \end{align*}
    where $\Phi$ is the CDF of the standard Normal  distribution function.
    
    \paragraph{Fast convergence} This case is very similar, except now $m(c_N)=\Theta(1)|c_N-c^*|=o(1/\sqrt N)$, which by Eq.~\eqref{eq:mu_sigma} means $\frac\mu\sigma=\sqrt{N}\sqrt{s^*}\cdot o(1/\sqrt N) = o(1),$
    i.e. in this case $\frac{\mu}{\sigma}$ goes to $0$ as $N$ grows, and thus
     \begin{align*}
  \mathbb{WP}_A(I,N,c_N)\cong  \Phi(\frac{\mu}{\sigma}) \xrightarrow[N\rightarrow \infty]{} 0.5.
    \end{align*}

    \paragraph{Moderate convergence}  In the knife-edge case where $cr(\ol c)=r(c^*)$ is a positive constant,  we have by Lemma~\ref{lemma:margin_bound} that $m(c_N)=\Theta(1/\sqrt N)$, and more precisely, that 
    $$m(c_N) = |c_N-c^*|(m^*+o(1)) = \frac{r(c^*)\cdot m^*}{\sqrt{N}}.$$

    This means that the ratio $\frac\mu\sigma$ also has a finite positive limit, specifically
    $$  \frac\mu\sigma=\sqrt{N}\sqrt{s^*}\cdot m(c_N)  =\sqrt{s^*}\cdot r(c^*)(m^*+o(1)),$$
    and by continuity of $\Phi$, we can define the constant $\phi(c^*):=\lim_{N\rightarrow \infty}\Phi(\frac{\mu}{\sigma}) = \Phi(\sqrt{s^*}\cdot r(c^*)m^*)$,
    and it holds that 
    $$ \mathbb{WP}_A(I,N,c_N)\cong  \Phi(\frac{\mu}{\sigma}) \xrightarrow[N\rightarrow \infty]{} \phi(c^*)<1,$$
    as required.
    \end{proof}

\section{Parameters of Specific PPMs}\label{apx:PPM_char}
\subsection{Polynomial PPM}
\lemNonTrivial*
That is, the point $c_N$ must be at distance that decreases proportionally to $N^{-\frac{\beta}{\alpha}}$: not closer neither farther away. 
\begin{proof}
We will prove for $c_N>c^*$. The proof for $c_N<c^*$ is symmetric.

    By Theorem~\ref{thm:nontrivial_eq}, $c^*$ is a pivot point. 

    Since the first and second derivatives of $s_A,s_B$ are bounded, so are those of the functions $n(c)$ and $m(c)$ (as per Eq.~\eqref{eq:n_c},\eqref{eq:m_c}).

    We denote $m'(c):=\frac{\partial m(c)}{\partial c}$ and $n'(c):=\frac{1}{N}\frac{\partial n(c,N)}{\partial c}$, so neither function depends on $N$.
    
    Also note that $n'(c)\geq 0$ everywhere and $m'(c^*)>0, n'(c^*)>0$ by the definition of pivot point. We also denote $s^*:=s_A(c^*)+s_B(c^*)$ and note that $s^*\cdot N = n(c^*)$.

    Therefore there must be some interval $[c^*,c^*+\delta]$ and constants $\ol m,\ul m,\ol n,\ul n>0$ such that $m'(c)\in [\ul m,\ol m]$ and $n'(c)\in[\ul n,\ol n]$ for all $c\in [c^*,c^*+\delta]$. Moreover, $\ul m,\ol m$ can be arbitrarily close to $m'(c^*)$ and likewise for $n'$ (the functions are nearly-linear near $c^*$).  We consider only $N$'s large enough so that $c_N$ is in this interval, and so that $p(n(c_N),m(c_N))<1$.

    Now, let $\eps:=c_N-c^*$, and w.l.o.g. $\eps<\min\{0.1,0.1/\ol n\}$. We want to show $\eps=\Theta(\frac{1}{N^\frac\beta\alpha})$. From the definition of equilibrium,
    \begin{align*}
        c^*+\eps &= c_N = p(n(c_N),m(c_N)) = \frac{1}{m(c_N)^\alpha n(c_N)^\beta}\\
        &=\frac{1}{m(c^*+\eps)^\alpha n(c^*+\eps)^\beta} \Rightarrow
    \end{align*}
    \begin{equation}\label{eq:m_c_star_bound}
        m(c^*+\eps) 
        = \frac{1}{(c^*+\eps)^\frac1\alpha n(c^*+\eps)^\frac\beta\alpha}.
    \end{equation}
        
    Now, by the constant bounds on the derivatives, 
    \begin{align}
        m(c^*+\eps) &\in [m(c^*)+\ul m \eps, m(c^*)+\ol m \eps] \label{eq:m_eps_bounds}\\
        &=  [\ul m \eps, \ol m \eps], \label{eq:m_u_o}\\
        n(c^*+\eps) &\in [n(c^*)+\ul n \eps\cdot N, n(c^*)+\ol n \eps\cdot N] \notag \\
        &= [(s^*+\ul n \eps)N, (s^*+\ol n\eps)N]\notag \\
        &\subset [s^* N, (s^*+0.1)N] \label{eq:s_star}
    \end{align}
For the upper bound, we get from Eqs.\eqref{eq:m_c_star_bound},\eqref{eq:m_u_o},\eqref{eq:s_star} 
\begin{align*}
    \ul m \eps &\leq m(c^*+\eps) \leq \frac1{(c^*)^\frac1\alpha (s^* N)^\frac{\beta}{\alpha}} \Rightarrow\\
    \eps & \leq \frac1{\ul m\cdot (c^*)^\frac1\alpha (s^*)\frac{\beta}{\alpha} N^\frac{\beta}{\alpha}}=O(\frac{1}{N^\frac{\beta}{\alpha}}),
\end{align*}
since $\ul m, c^*$ and $s^*$ are all constants. 
Likewise, for the lower bound, 
\begin{align*}
     \ol m \eps &\geq m(c^*+\eps) \geq \frac1{(c^*+0.1)^\frac1\alpha ((s^*+0.1) N)^\frac{\beta}{\alpha}} \Rightarrow\\
     \eps &= \Omega(\frac{1}{N^\frac{\beta}{\alpha}}) \text{ as required.}
\end{align*}
\end{proof}
Clearly the term ``$0.1$" used in the proof is arbitrary and could be replaced by a smaller constant. Thus the lemma shows 
$$c^* + m(c^*)^{-1}\cdot \ul t \cdot N^{-\frac\beta\alpha}< c_N < c^* + m(c^*)^{-1}\cdot \ol t \cdot N^{-\frac\beta\alpha},$$
where $m'$ is the derivative of $m(c)$ at $c^*$. 
Moreover, from Eq.~\eqref{eq:m_eps_bounds},
$$ \ul t\cdot N^{-\frac\beta\alpha}<m(c_N) < \ol t\cdot   N^{-\frac\beta\alpha},$$
which we use in the proof of Theorem~\ref{thm:JT}. In the limit as $N$ grows, both $\ol t, \ul t$ converge to \begin{equation}\label{eq:t_bounds}
    t^*=(c^*)^{-\frac1\alpha} (s^*)^{-\frac\beta\alpha},
\end{equation}
and thus $r^* = t^*\cdot \sqrt{s^*} =(c^*)^{-\frac1\alpha} (s^*)^{0.5-\frac\beta\alpha}$.

\subsection{Altruist PPM}
\PropAltruistic*

Note that there is no $f$ for which there is a non-trivial equilibrium with fast convergence. 
\begin{proof}
  Case~1. To see that $c=q$ is an equilibrium, note that  the margin at $m(q)$ is strictly less than 1 (say, $>1-\delta$) and this $p^{Bin}(n,m(q))>\delta^n$. 
  Also, $f(n)>e^{Z\cdot n}$ for any constant $Z$ and in particular for $Z=-\log \delta$.
  
  Thus 
  \begin{align*}
  p^{alt(q,f)}(n,m(q))&=\min\{q,f(n)\delta^n\}>\min\{q,e^{-\log\delta\cdot n}\delta^n\}\\
  &=\min\{q,\delta^{-n}\delta^n\}=\min\{q,1\}=q.
  \end{align*}
  That is, all voters with cost $c\leq q$ consider themselves pivotal.

    \medskip
    For Cases~2 and 3, we know from Lemma~\ref{lemma:alt_vanish} and Theorem~\ref{thm:nontrivial_eq} that $\ol c$ converges to a pivot point $c^*\in(0,1)$.  
    
    For Case~2, assume towards a contradiction that convergence is (at least) moderate, i.e., there is a constant $X>0$ s.t.  $|c_N-c^*| < X/\sqrt{N}$ for all $N$.

     Now, we have that in equilibrium $c=c_N$ with $n:=n(c_N,N)$ and $$\eps:=m(c_N)\leq \sup_c m'(c) \cdot |c_N-c^*|\leq Y /\sqrt{N},$$
     for some constant $Y$,
     due to bounded derivatives.

      Set $Z:=e^{Y^2}$.
 Since $f(n)=\omega(\sqrt n)$, for any sufficiently large $n$, we have $f(n)>4Z\sqrt{n}$.
 Note that the number of active voters at the pivot point $c^*$ is $s^*\cdot N$, and that $n(c_N)= s_N\cdot N \cong s^*N$. 
 
     Note that for all $n$,
     \begin{align*}
         p^{Bin}&(2n+1,2\eps)=\binom{2n}{n}(\frac12+\eps)^n(\frac12-\eps)^n \\ & > \frac{4^n}{2\sqrt{\pi n}} (\frac14-\eps^2)^n \\
         &>\frac{1}{\sqrt n}(1-\eps^2)^n >\frac{1}{\sqrt n}(1-(Y/\sqrt N)^2)^n \\
         &= \frac{1}{\sqrt n}(1- Y^2/N)^n\geq \frac{1}{\sqrt n}(1-\frac{Y^2}{n})^n \\ &>\frac{1}{2\sqrt n} e^{-Y^2}
     \end{align*}
     for sufficiently large $n$.

    Finally, pick $N$ large enough so that $n=n(c_N)$ respects the conditions above, then
     \begin{align*}
         1&>q\geq  c_N=p^{alt(q,f)}(n,\eps)=p^{Bin}(n,\eps)f(n) \\
         &> p^{Bin}(n,\eps) Z\sqrt n\\
         &>\frac{1}{2\sqrt{n/2}} e^{-Y^2} \cdot 4Z\sqrt n\\
         &>e^{-Y^2} \cdot Z = 1,
         \end{align*}
         which is a contradiction. 

   \medskip
   Case 3. We start from
$p^{Bin}(n,\eps)= O(1)\frac{1}{\sqrt n}(1-\eps^2)^n$, which means $p^{alt(f)}(n,\eps) = O(1)(1-\eps^2)^n$.
Since this converges to some pivot point $c^*$, we have that 
$$(1-\eps^2)^n\rightarrow Z,$$
for some constant $Z<1$.
Thus the margin at equilibrium holds
$$\eps=\eps(C_N) \cong \sqrt{1-Z^{1/n}} = \sqrt{\frac1n\log\frac1Z}=\Theta(\sqrt{1/N}),$$
where we use the fact that  $\lim_{n\rightarrow \infty}(1-Z^{1/n})n=\log\frac1Z$.
Finally, the gap $|c_N-c^*|$ is linear in the margin $\eps$ due to bounded derivatives. 

\medskip
  Case~4 follows immediately from Lemma~\ref{lemma:alt_vanish} and Prop.~\ref{prop:strong_trivial}, as the PPM is strongly vanishing. 
 \end{proof}

\section{Partially-coinciding Support Functions}\label{apx:overlap}

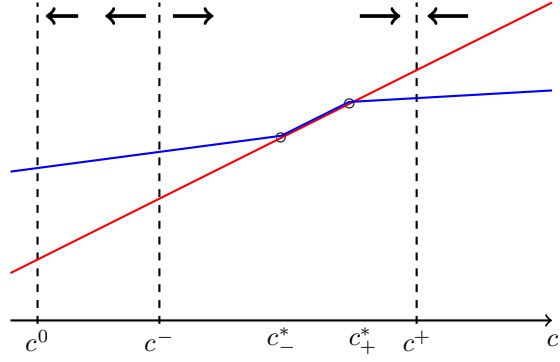
\begin{figure}[t!]
 \centering 
 \begin{tikzpicture}[scale=0.9, blend group=soft light]
\def\a{40}
\def\bx{-20}
\def\by{-13}
%\draw[pattern=south west lines, pattern color=red] plot[smooth,samples=100,domain=0:5.55] (\x,{3}) -- 
 %   plot[smooth,samples=50,domain=5.55:0] (\x,{0.3});
%\draw[pattern=north west lines, pattern color=blue] plot[smooth, samples=100,domain=0:5.55] (\x,{ 0.97 + 0.5*\x}) -- plot[smooth,samples=100,domain=5.55:5.55] (\x,{0.25})  -- plot[smooth,samples=100,domain=5.55:0] (\x,{0.28})  ;
\draw[blue,thick] (0,2.5) -- (4,3.03) -- (5,3.53) -- (8,3.7);
\draw[red,thick] (0.5*\a+\bx, 0.35*\a+\by) -- (0.7*\a+\bx, 0.45*\a+\by);
\draw[thick,->] (0,0.3) -- (8,0.3);
\draw[dashed,thick] (0.65*\a+\bx,0.25) -- (0.65*\a+\bx,5);
\draw[dashed,thick] (0.555*\a+\bx,0.25) -- (0.555*\a+\bx,5);
%\draw[dashed,thick] (0.6*\a+\bx,0.25) -- (0.6*\a+\bx,5);
\draw[dashed,thick] (0.51*\a+\bx,0.25) -- (0.51*\a+\bx,5);
\draw[double,->] (0.55*\a+\bx,4.8) -- (0.535*\a+\bx,4.8);
\node at (0.55*\a+\bx + 0.6,4.8 +0.4 ) {$ $};
\node at (0.55*\a+\bx - 0.31,4.8 +0.4 ) {$ $};
\draw[double,->] (0.56*\a+\bx,4.8) -- (0.575*\a+\bx,4.8);
\draw[double,->] (0.629*\a+\bx,4.8) -- (0.644*\a+\bx,4.8);
\draw[double,->] (0.669*\a+\bx,4.8) -- (0.654*\a+\bx,4.8);
\draw[double,->] (0.525*\a+\bx,4.8) -- (0.513*\a+\bx,4.8);

%\draw[dotted,thick] (0.639*\a+\bx,  0.418*\a+\by) -- (0, 0.418*\a+\by);
\node at (0.6*\a+\bx, 0.4*\a+\by) {$\circ$};

\node at  (5,3.5) {$\circ$};
%\node at (0.628*\a+\bx,-0) {$0.628$};
\node at (0.65*\a+\bx,-0) {$c^+$};
%\node at (0.569*\a+\bx,-0) {$0.569$};
\node at (8,0) {$c$};
%\draw[draw=gray!50!white,fill=gray!50!white] 
%    plot[smooth,samples=100,domain=0:1] (\x,{0}) -- 
 %   plot[smooth,samples=100,domain=1:0] (\x,{1});
\node at (0.555*\a+\bx,-0) {$c^-$};
\node at (0.51*\a+\bx,-0) {$c^0$};

%\node at (0.6*\a+\bx,-0) {$0.6$};
\node at (0.6*\a+\bx,-0) {$c^*_-$};
\node at (0.63*\a+\bx,-0) {$c^*_+$};

%\node at (-0.4, 0.4*\a+\by) {$0.4$};
%\node at (-0.5, 0.35*\a+\by) {$0.35$};
%\node at (-0.6, 0.414*\a+\by) {$0.414$};
%\node[blue] at (0.7*\a+\bx, 0.4*\a+\by +0.3) {$s_B = 0.4$};
%\node[red] at (0.7*\a+\bx, 0.45*\a+\by+0.3) {$s_A = 0.1 + \frac{c}{2}$};

% \tikzstyle{dot}=[rectangle,draw=black,fill=white,inner sep=0pt,minimum size=4mm]
% \tikzstyle{active}=[circle,draw=blue,fill=blue,inner sep=0pt,minimum size=2mm]
% \tikzstyle{inactive}=[circle,draw=blue,fill=white,inner sep=0pt,minimum size=2mm]
% \tikzstyle{sybil}=[rectangle,draw=black,fill=red,inner sep=0pt,minimum size=1.6mm]
% \tikzstyle{del}=[triangle,draw=blue,fill=white,inner sep=0pt,minimum size=2mm]
% \tikzstyle{virtual}=[diamond,draw=black,fill=gray,inner sep=0pt,minimum size=2mm]
% \tikzstyle{txt}=[text width = 8cm, anchor=west]

% \node at (2,0.5) {$r$};
% \node at (2,0) {$*$};
% \node at (5,0) {$*$};
% %\node at (5,-0.5) {$\calG(H)$};
% %\node at (4.4,1.35) [txt] {$\overline \calG_\beta(H)$};

% \draw[thick] (3.6,0) -- (5.8,0);
% \draw[dashed] (3.6,0) -- (2,0);

% \draw [decorate,decoration={brace,amplitude=4pt,raise=4pt},yshift=0pt]
% (3.6,0.1) -- (5.8,0.1) node [black,midway,yshift=0.8cm] {
% $\overline \calG_\safe(H)$};

% \draw [decorate,decoration={brace,amplitude=4pt,raise=4pt},yshift=0pt]
% (5.8,-0.1) -- (2,-0.1) node [black,midway,yshift=-0.8cm] {
% $\calB(r;\overline \calG_\safe(H))$};
\end{tikzpicture}
\caption{Pivot points and equilibria when the support functions are partially overlapping.
}
\label{fig:overlap_support}
\end{figure}

The assumption of a finite number of pivot points for $s_A,s_B$ is somewhat restrictive. Suppose there are segments of the cost distribution where voters are equally likely to prefer $A$ or $B$: this would translate to intervals $[c^*_-,c^*_+]$ where $s_A(c)=s_B(c)$ for all $c\in [c^*_-,c^*_+]$.  See Fig.~\ref{fig:overlap_support}.  By our assumption on the introduction the support functions should be differentiable near the pivot points (so the figure is not allowed), but we actually only need them to have bounded derivative in some \emph{open right environment} of $c^*_+$, and likewise for the left pivot point. In other words, the left- and right- derivative at $c^*$ could be different but has to exist.

The effect on our result in Theorem~\ref{thm:nontrivial_eq} is that the right equilibrium will converge to $c^*_+$ (from the right), and the left equilibrium will converge to $c^*_-$ from the left. The proof itself remains unchanged. 

Can points on the interval $[c^*_-,c^*_+]$ be an equilibrium? note that since $m(c)=0$ for all $c\in[c^*_-,c^*_+]$, we have $p(n,m(c))=p(n,0)>q$ by $q$-tie-sensitivity, and for sufficiently large $N$ we should get $p(n(c,N),0)=q$.  We conclude that if the interval contains $q$ then $q$ will also be an equilibrium. Also note that this equilibrium will be stable, as with any small perturbation all voters still think they are $q$-pivotal.

An assumption we cannot relax is that the margin at the [right side of the right-] pivot point must be strictly increasing. If the support functions have the same asymptote at $c^*$ then equilibrium will still exist but we cannot say anything about the rate of convergence of $c_N$ to $c^*$, or on the winning probability.

\section{Stability}\label{apx:stable}
Recall that every agent  has two (undominated) actions: vote/abstain. A Profile is a mapping from the nonatomic set of agents to a subset of active voters.  So for a given election, we can think of every profile $P$ as a subset of agents. Every agent in profile $P$ has a best response (to $P$), which is either to keep her action or switch to the other action. We thus denote by $BR(P)$ the profile obtained from $P$ if all voters simultaneously switch to their best-response. Note that the Bayes-Nash equilibria are those profiles where $BR(P)=P$.

\begin{definition} Let $P, Q$ be two  profiles. 
    \begin{enumerate}
        \item Write $d(P,Q)$ to denote  their \emph{difference}, i.e. the fraction of agents playing differently. 
        \item 
The \emph{$\eps$-neighborhood} of a profile $P$ is the set of profiles with difference at most $\eps$.
    \end{enumerate}
\end{definition}

Recall that $P$ is a \emph{threshold profile} if there is $c$ such that exactly all voters with $c_i\leq c$ vote in $P$. In such a case we denote the threshold by $c(P)$.

An easy observation is that for any $P$, $BR(P)$ is a threshold profile.  This is since in $P$ there is some expected number of active voters $n_P$ and some expected margin $m_P$ and thus the threshold of $BR(P)$ would be $c(BR(P))=p(n_P,m_P)$.

In particular, $BR(P)$ for some threshold profile with $c(P)=a$ is also a threshold profile with some $c(BR(P))=b$. We then denote $b=BR(a)$. 

\begin{definition}
    We say that a threshold profile $P$ is \emph{moving towards} a threshold $q$ if both $q$ and $c(BR(P))$ are on the same side of $c(P)$.
\end{definition}

\begin{definition}[Stable equilibrium]
    An election equilibrium $Q$ with threshold $c_N$ is \emph{stable}, if there is some $\eps>0$ such that for any profile $P$ in a $\eps$-neighborhood of $Q$:
    \begin{itemize}
        \item If $P$ is a threshold profile then it is moving towards $c_N$;
        \item Otherwise, $BR(P)$ (which is a threshold profile) is moving towards $c_N$.
    \end{itemize}
\end{definition}
This means that when starting from profile $P$, a best response by few agents will get closer to $Q$, so it will still be in the $\eps$-neighborhood of $Q$ and will continue to get closer in every step until convergence. 

\begin{lemma}\label{lemma:p_mon}There is some interval  $Z=[c^*,c^*+z]$  where for every $N$, $p(n(c,N),m(c))$ is monotonically decreasing in $c\in Z$. Moreover, for every $\eps>0$ there is sufficiently large $N_\eps$ such that 
 \begin{enumerate}
     \item $p(n(c,N_\eps),m(c))$ is strictly decreasing in $c$ in the range $c\in [c^*+\eps,c^*+z]$;
     \item $c_{N_\eps}\in (c^*+\eps,c^*+z/2]$.
 \end{enumerate}
  \end{lemma}
 \begin{proof}
 Note that $n(c,N)$ is non-decreasing in $c$ everywhere. In addition, $m(c^*)=0$, and by our assumption that $s_A,s_B$ have different derivatives at $c^*$, we get  $m'(c^*)>0$. So from our bounded derivative assumption $m$ must be strictly increasing in some range $[c^*,c^*+z]$ that is independent of $N$. Since $p(n,m)$ is monotonically non-increasing in both parameters, $p$ is non-increasing in $c$ in the entire range $[c^*,c^*+z]$, and this holds for every $N$.

Showing that it strictly decreasing is not immediate because in principle it could equal the maximum value in the entire range. Indeed this is the case for the Polynomial PPM if $N$ is low.
However, from weak vanishing pivotality, we know that  $p(n(c^*+\eps,N), m(c^*+\eps))$ goes to 0 with $N$, and thus drops below the maximal value for sufficiently large $N_\eps$. From our assumption on $p(n,m)$, it is strictly decreasing in both parameters below its maximal value, and thus strictly decreasing in $c$ in the range $(c^*+\eps,c^*+z]$ whenever $N\geq N_\eps$. 

For the minimal $\eps$ for which this holds, we have $p((n(c^*+\eps,N_\eps), m(c^*+\eps))=1>c_N$ and so $c_N$ must be more to the right, where the pivotality decreases. 

We increase $N$ further as needed until $c_N$ is in the range $(c^*+\eps, c^*+z/2)$ (must occur eventually as the limit of $(c_N)_N$ is $c^*$).
 \end{proof}
 The next definition generalizes Def.~\ref{def:stable_thresh}: 
 \begin{definition}[Stable equilibrium]\label{def:stable}

An election equilibrium $Q$ with threshold $c_N$ is stable, if there is some $\varepsilon > 0$ such that for
any profile $P$ in a $\varepsilon$-neighborhood of $Q$:
\begin{itemize}
    \item  If $P$ is a threshold profile then it is moving towards $c_N$ ;
    \item Otherwise, $BR(P)$ (which is a threshold profile) is moving towards $c_N$. 
\end{itemize}
\end{definition}

\stableEq*
\begin{proof}
%[Proof of Proposition~\ref{prop:stable_eq}] 
We first explain how to select $N$.
Suppose that $N$ is sufficiently large so that by Lemma~\ref{lemma:p_mon}, $p$ is strictly decreasing in $c$ in some range $(c^*+\eps',c^*+z]$ and $c_N \in (c^*+\eps',c^*+z/2]$. 

\medskip
Consider a profile $P$ in the $\eps$-neighborhood of $c_N$, where $\eps$ is determined below.

\noindent\textbf{Threshold profile:} Suppose first that $P$ is a threshold profile with threshold $\hat c$. Then choose $\eps$ small enough so that $|c_N-\hat c|<\min\{c_N-(c^*+\eps'),z/2\}$.
 This guarantees that:
 \begin{align*}
     \hat c &> c_N- (c_N-(c^*+\eps') = c^*+\eps',\\
     \hat c&<c_N+z/2< c^*+z.
 \end{align*}

    Since both $\hat c, c_N$ are in the interval $[c^*+\eps',c^*+z]$ where $p$ is strictly decreasing, when $\hat c<c_N$:
    $$BR(\hat c) = p(n(\hat c),m(\hat c))>p(n(c_N),m(c_N))=c_N>\hat c,$$
so both $c_N, BR(\hat c)$ are above $\hat c$. Similarly, when $\hat c>c_N$ then both would be below. 

\medskip
\noindent\textbf{Non-threshold profile:}
Next, suppose that $P$ is \emph{not} a threshold profile. Then we only need to show that $\hat c= c(BR(P))$ is also in the interval $[c^*+\eps',c^*+z]$.

Denote $\eps:=(c_N-(c^*+\eps'))/t$, where $t$ will be determined later.

Denote by $\hat n, \hat m$ the expected number of voters and expected margin in the current profile $P$. Note that
      $$\hat n \in (1-\eps)(s_A(c_N)+s_B(c_N))N \pm \eps\cdot N\subseteq  n(c_N,N)\pm \eps$$

 and similarly
     $$\hat m \in \frac{|s_A(c_N)-s_B(c_N)|\pm \eps}{s_A(c_N)+s_B(c_N)\pm \eps} \subseteq  m(c_N)\pm 2\eps\cdot s',$$
     where $s'=s_A(c_N)+s_B(c_N)$ is a constant. 
      
    Now set $\hat c:=p(\hat n,\hat m)$ and note that since $p$ has bounded derivatives, 
    $$\hat c =p(\hat n,\hat m) \in p(n(c_N,N),m(c_N)) \pm O(\eps) = c_N\pm h\cdot\eps,$$
    for some constant $h$.
    We now set $t>h$, so that  
    $$\hat c> c_N-h\cdot \eps > c_N-t\cdot \eps=c_N-(c_N-(c^*+\eps'))  = c^*+\eps',$$
    and
    \begin{align*}
        \hat c &< c_N+h\cdot \eps \\ &<  c_N+(c_N-c^*)
        \\ & =c^*+2(c_N-c^*) \\ &<c^*+2(z/2) \\&=c^*+z.
    \end{align*}
Thus $\hat c$ (which is the threshold of the best-response to profile $P$) is still on the monotone interval and the previous part of the proof shows it is moving towards $c_N$.

% Clearly the best-response of voters with $c_i\leq \hat c$ is to vote, and for those with $c_i>\hat c$ is to abstain. 

% It remains to show that the overall trend is towards $c_N$. I.e. that if $\hat c<c_N$ then more voters would choose to switch to voting from abstaining than vice versa, and the opposite if $\hat c>c_N$.

%     Since both $\hat c, c_N$ are in the interval $[c^*+\eps,c^*+z]$ where $p$ is strictly decreasing, when $\hat c<c_N$:
%     $$p(n(\hat c),m(\hat c))>p(n(c_N),m(c_N))=c_N>\hat c,$$
% so indeed more voters are incentivized to switch from abstaining to voting. When $\hat c>c_N$ the inequalities flip and more voters abstain. 
 \end{proof}

\section{Calculus of Voting PPM}\label{apx:CoV}
 \begin{example}[CoV PPM] The Calculus  of Voting PPM (note it also depends on the size of the entire population $N$)
 \begin{align}
  \ppm^{CoV}&(n,m,N) := E_{n' \sim Bin(N,n/N)}[p^{Bin}(n',m)] \nonumber \\ & =   E_{n' \sim Bin(N,n/N)}[\Pr_{x\sim \text{Bin} \big ( n', (1+m)/2 \big)}(x = \floor{n'/2})]. 
  \end{align}
 \end{example}

 Note that from Equations~\eqref{eq:n_c},\eqref{eq:m_c} we can also derive the expected fraction of A and B supporters from $n,m$ and $N$, as 
 $$s_A,s_B=\frac{n^2 \pm mN^2}{2n^2} = \frac12 \pm (\frac{N}{n})^2m.$$
 \begin{proposition}
     The following three terms coincide:
     \begin{enumerate}
         \item The probability of a tie, i.e. $Pr(V_A=V_B|I,N,c)$;
         \item $\ppm^{CoV}(n,m,N)$;
         \item $Pr_{(V_A,V_B,V_0)\sim Mult(N,(s_A,s_B,1-s_A-s_B))}[V_A=V_B]$.
     \end{enumerate}
 \end{proposition}
 \begin{proof}
     There is a tie if there is exactly the same number of active $A$ voters and $B$ voters. One way to compute this is to first sample $n'$ active voters by flipping a coin for each of the $N$ voters. Then for each of the $n'$ active voters decide (w.p. $s_A$ vs. $s_B$) if she is an $A$ or $B$ supporter. There is a tie if there are exactly $n'/2$ A supporters, which happens w.p. $p^{Bin}(n',m)$.

     Alternatively, we could just decide for each of the $N$ voters whether she is an active A supporter (which occurs w.p. $s_A = Pr_{(c_i,T_i)\sim \calD}[i \text { is active and supports }A]$), an active $B$ supporter, or inactive. Then we check if the two (correlated) multinomial variables $V_A$ and $V_B$ are the same.  
 \end{proof}
\begin{proposition}
    The CoV PPM has strongly vanishing pivotality. I.e., for every fixed $c$, $\lim_{N\rightarrow \infty}Pr(V_A=V_B|I,N,c)=0$. Moreover,
    \begin{itemize}
        \item $p^{CoV}(n,m,N) =\Theta(\frac{1}{\sqrt{n}})$ if $m=0$; and
        \item $p^{CoV}(n,m,N) = exp(-\Theta(n))$ if $m>0$.
    \end{itemize}
    (assuming the ratio $n/N$ remains fixed)
\end{proposition}
\begin{proof}
    As $N$ goes to infinity, w.h.p we have $n'\in[0.99n,1.01n]$. Also note that $m$ remains fixed. Thus 
    $$Pr(V_A=V_B|I,N,c)\leq p^{Bin}(0.99n,m)+exp(\Theta(-N)),$$
    and
    $$Pr(V_A=V_B|I,N,c)\geq p^{Bin}(1.01n,m)-exp(\Theta(-N)).$$
    Now we have $p^{Bin}(n,m)=\Theta(\frac{1}{\sqrt{n}})$ for $m=0$ and  $p^{Bin}(n,m)=exp(-\Theta(n))$ for $m>0$, and changing $n$ by a constant fraction does not change the asymptotic result. 
\end{proof}

\subsection{Poisson PPM}
\begin{example}
  [Poisson PPM]% Let $\mathcal{I}$ be a given election instance.   
  The Poisson PPM  considers the perceived pivotality as the probability that an equal number of supporters are drawn from Poisson distributions with parameters $N \cdot s_A $ and $N \cdot s_B$:
  \begin{equation}
      \ppm^{Poi}(n,m):=\Pr_{\substack{X_A\sim \text{Poisson}((1+m)n/2)\\ X_B\sim \text{Poisson} ((1-m)n/2)}}(X_A=X_B)\label{eq:Poisson}
  \end{equation}
 \end{example}
Note that while $s_A, s_B$ and $N$ cannot be inferred from $n,m$,  the Poisson parameters can be inferred as $s_A\cdot N= \frac{s_A}{s_A+s_B}n=\frac{1+m}{2}n$, and $s_B\cdot N=\frac{1-m}{2}n$.
 
\section{Diverse PPM}\label{apx:diverse}
Suppose that each type includes not just voting cost $c_i$ and preferred candidate $T_i$, but also $i$'s own pivotality estimation function $p_i(n,m)$, which we still assume to be continuous and decreasing in both parameters. 

We fix an election $(I,N)$, where $I$ is a distribution over types.

\begin{proposition}
    Every election $(I,N)$ has a pure Bayes-Nash equilibrium. 
\end{proposition}
\begin{proof}
The functions $s_A(c)$ and $s_B(c)$ become useless, since there is no meaningful cost threshold that applies to all agents. 

Also let $a,b$ be the fraction of core supporters of $A$ and $B$, respectively. 
However, we will maintain the notation $s_A,s_B$
for the fraction of all voters who actively vote for $A$ and $B$, respectively. 

Note that 
$$(s_A,s_B)\in \Delta := \{(x,y) : (x\geq a) \wedge (y\geq b) \wedge (x+y\leq 1)\}.$$

Given a pair of numbers $s_A, s_B$, we can still compute the expected number of active voters and expected margin as in Eq.~\eqref{eq:n_c},\eqref{eq:m_c} (recall that $N$ is fixed):

\begin{align}
    n(s_A,s_B)&:=(s_A+s_B)N \label{eq:n_s}; &\text{(active voters)}\\
    m(s_A,s_B)&:=\frac{|s_A-s_B|}{s_A+ s_B}\label{eq:m_s}.&\text{(expected margin)}
\end{align}

We now define functions $S_A,S_B$ which map $n$ and $m$ to $(s_A,s_B)\in \Delta$, by integrating over the distribution $I$:
\begin{align*}
    S_A(n,m) &:= Pr_{i\sim I}[c_i \leq p_i(n,m) \wedge T_i=A];\\
    S_B(n,m) &:= Pr_{i\sim I}[c_i \leq p_i(n,m) \wedge T_i=B];\\
\end{align*}

Finally we get a function $F:\Delta\rightarrow\Delta$ defined as
$$F(s_A,s_B):=\left(S_A(n(s_A,s_B),m(s_A,s_B)), S_B(n(s_A,s_B),m(s_A,s_B))\right).$$
Since $I$ is atomless, $p_i$ are continuous, and $n(), m()$ are continuous,  so are $S_A()$ and $S_B()$. Thus $F$ is a continuous function from a compact and convex set $\Delta$ onto itself. From Brouwer's fixed point theorem, $F$ has a fixed point. This point is a pure Nash equilibrium of $(I,N)$.
\end{proof}
\section{Additional Insights from Simulations}

 \begin{restatable}{example}{exOne}
Suppose $s_A(c)=0.1+c/2$ and $s_B(c)=0.4$. That is, $B$ has 40\% overall support, all of them core supporters and $A$ has 60\% overall support and a fraction of  voters are distributed  over cost range $[0,1]$. We assume the Tie-sensitive PPM with $\alpha=1$. 
\label{ex:first}
\end{restatable}

The code used for the simulation results is available in anonymized repository:

\url{https://anonymous.4open.science/r/NJT-E778/}.

 Our goal here is to see how $A$'s winning probability changes as population size increases and/or when we vary the parameters. For simplicity and easy computation we use the Polynomial PPM with $\beta=\frac12$.

 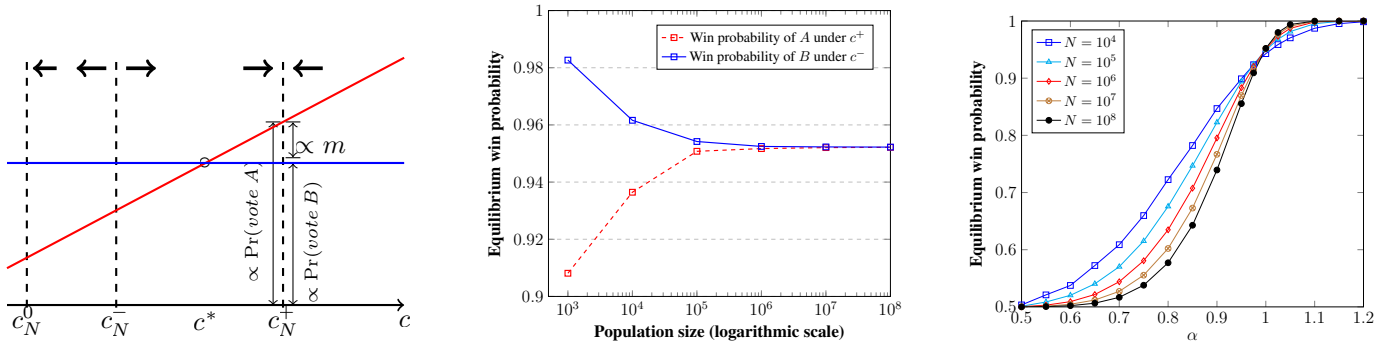
\begin{figure*}[ht]
\centering
\begin{subfigure}{0.28\textwidth}
     \centering 
 \begin{tikzpicture}[xscale=0.66, yscale=0.7, blend group=soft light]
\def\a{40}
\def\bx{-20}
\def\by{-13}
%\draw[pattern=south west lines, pattern color=red] plot[smooth,samples=100,domain=0:5.55] (\x,{3}) -- 
 %   plot[smooth,samples=50,domain=5.55:0] (\x,{0.3});
%\draw[pattern=north west lines, pattern color=blue] plot[smooth, samples=100,domain=0:5.55] (\x,{ 0.97 + 0.5*\x}) -- plot[smooth,samples=100,domain=5.55:5.55] (\x,{0.25})  -- plot[smooth,samples=100,domain=5.55:0] (\x,{0.28})  ;
\draw[blue,thick] (0.5*\a+\bx, 0.4*\a+\by) -- (0.7*\a+\bx, 0.4*\a+\by);
\draw[red,thick] (0.5*\a+\bx, 0.35*\a+\by) -- (0.7*\a+\bx, 0.45*\a+\by);
\draw[thick,->] (0,0.3) -- (8,0.3);
\draw[dashed,thick] (0.639*\a+\bx,0.25) -- (0.639*\a+\bx,5);
\draw[dashed,thick] (0.555*\a+\bx,0.25) -- (0.555*\a+\bx,5);
%\draw[dashed,thick] (0.6*\a+\bx,0.25) -- (0.6*\a+\bx,5);
\draw[dashed,thick] (0.51*\a+\bx,0.25) -- (0.51*\a+\bx,5);
\draw[double,->] (0.55*\a+\bx,4.8) -- (0.535*\a+\bx,4.8);
\node at (0.55*\a+\bx + 0.6,4.8 +0.4 ) {$ $};
\node at (0.55*\a+\bx - 0.31,4.8 +0.4 ) {$ $};
\draw[double,->] (0.56*\a+\bx,4.8) -- (0.575*\a+\bx,4.8);
\draw[double,->] (0.619*\a+\bx,4.8) -- (0.634*\a+\bx,4.8);
\draw[double,->] (0.659*\a+\bx,4.8) -- (0.644*\a+\bx,4.8);
\draw[double,->] (0.525*\a+\bx,4.8) -- (0.513*\a+\bx,4.8);

%\draw[dotted,thick] (0.639*\a+\bx,  0.418*\a+\by) -- (0, 0.418*\a+\by);
\node at (0.6*\a+\bx, 0.4*\a+\by) {$\circ$};

%\node at (0.628*\a+\bx,-0) {$0.628$};
\node at (0.639*\a+\bx,-0) {$c^+_N$};
%\node at (0.569*\a+\bx,-0) {$0.569$};
\node at (8,0) {$c$};
%\draw[draw=gray!50!white,fill=gray!50!white] 
%    plot[smooth,samples=100,domain=0:1] (\x,{0}) -- 
 %   plot[smooth,samples=100,domain=1:0] (\x,{1});
\node at (0.555*\a+\bx,-0) {$c^-_N$};
\node at (0.51*\a+\bx,-0) {$c^0_N$};

%\node at (0.6*\a+\bx,-0) {$0.6$};
\node at (0.6*\a+\bx,-0) {$c^*$};
%\node at (-0.4, 0.4*\a+\by) {$0.4$};
%\node at (-0.5, 0.35*\a+\by) {$0.35$};
%\node at (-0.6, 0.414*\a+\by) {$0.414$};
%\node[blue] at (0.7*\a+\bx, 0.4*\a+\by +0.3) {$s_B = 0.4$};
%\node[red] at (0.7*\a+\bx, 0.45*\a+\by+0.3) {$s_A = 0.1 + \frac{c}{2}$};
\draw[|<->|] (0.639*\a+\bx+0.2, 0.4*\a+\by+0.09) -- (0.639*\a+\bx+0.2, 0.414*\a+\by + 0.23 );
\draw[<->|] (0.639*\a+\bx+0.2, 0.3) -- (0.639*\a+\bx+0.2, 0.4*\a+\by + 0.01 );
\draw[<->|] (0.639*\a+\bx-0.2, 0.3) -- (0.639*\a+\bx-0.2, 0.414*\a+\by + 0.23 );
\begin{scriptsize}
\node[rotate=90] at (0.639*\a+\bx+0.6,1.5) {$ \propto \Pr({vote}~B)$};
\node[rotate=90] at (0.639*\a+\bx-0.6,1.9) {$ \propto \Pr({vote}~A)$};
\end{scriptsize}
\node at (0.639*\a+\bx+0.75, 0.4*\a+\by+0.3) {$\propto m$};
% \tikzstyle{dot}=[rectangle,draw=black,fill=white,inner sep=0pt,minimum size=4mm]
% \tikzstyle{active}=[circle,draw=blue,fill=blue,inner sep=0pt,minimum size=2mm]
% \tikzstyle{inactive}=[circle,draw=blue,fill=white,inner sep=0pt,minimum size=2mm]
% \tikzstyle{sybil}=[rectangle,draw=black,fill=red,inner sep=0pt,minimum size=1.6mm]
% \tikzstyle{del}=[triangle,draw=blue,fill=white,inner sep=0pt,minimum size=2mm]
% \tikzstyle{virtual}=[diamond,draw=black,fill=gray,inner sep=0pt,minimum size=2mm]
% \tikzstyle{txt}=[text width = 8cm, anchor=west]

% \node at (2,0.5) {$r$};
% \node at (2,0) {$*$};
% \node at (5,0) {$*$};
% %\node at (5,-0.5) {$\calG(H)$};
% %\node at (4.4,1.35) [txt] {$\overline \calG_\beta(H)$};

% \draw[thick] (3.6,0) -- (5.8,0);
% \draw[dashed] (3.6,0) -- (2,0);

% \draw [decorate,decoration={brace,amplitude=4pt,raise=4pt},yshift=0pt]
% (3.6,0.1) -- (5.8,0.1) node [black,midway,yshift=0.8cm] {
% $\overline \calG_\safe(H)$};

% \draw [decorate,decoration={brace,amplitude=4pt,raise=4pt},yshift=0pt]
% (5.8,-0.1) -- (2,-0.1) node [black,midway,yshift=-0.8cm] {
% $\calB(r;\overline \calG_\safe(H))$};
\end{tikzpicture}
 
%\subcaption{The pivot point $c^*$ (circle at the intersection) with the two non-trivial equilibria for some specific $N$ on its sides (dashed lines), and the trivial equilibrium $c^0$. For $c^+_N$, the probability of a random voter to vote $A$ is proportional to $s_A(c^+_N)$. The $m(c^+_N)$ is   proportional to the margin of victory. The bold arrows  indicate that $c^+_N, c^0_N$ are stable  equilibria whereas $c^-_N$ is not stable.
%}
\label{fig:linearSupport_apx}
\end{subfigure}
\hfill 
\begin{subfigure}{0.28\textwidth}

    \centering
 \begin{tikzpicture}[scale=0.54] 
\begin{axis}[
    xmode=log,
    log ticks with fixed point,
    xlabel={\textbf{Population size (logarithmic scale)}},
    ylabel={\textbf{Equilibrium  win probability}},
    xmin=500, xmax=10^8,
    ymin=0.9, ymax=1,
    label style={font=\Large},
    tick label style={font=\Large}, 
    xtick={10^3, 10^4, 10^5, 10^6,10^7,10^8},
    xticklabels = {$10^3$,$10^4$, $10^5$, $10^{6}$,$10^7$,$10^8$},
    ytick={0.88,0.9,0.92,0.94,0.96,0.98,1},
    legend pos=north west,
    ymajorgrids=true,
    grid style=dashed,
    legend style={
        at={(axis description cs:0.95,0.95)},
        anchor=north east
    }
]
\addplot[thick, 
    color=red,
    mark=square, 
    mark options={solid},  
    dashed
]
    coordinates {
        (10^3, 0.90810929)(10^4, 0.93644931)(10^5, 0.95076779)(10^6, 0.95171826)(10^7, 0.95206293)(10^8, 0.95216058)
    };
\addplot[thick,
    color=blue,
    mark options={solid},
    mark=square
]
    coordinates {
        (10^3, 0.98266767)(10^4, 0.96161482)(10^5, 0.95419222)(10^6, 0.95247038)(10^7, 0.95230315)(10^8, 0.95223574)
    };
\legend{Win probability of $A$ under $c^{+}$, Win probability of $B$ under $c^-$}
\end{axis}
\end{tikzpicture}

%    \caption{Win probability for different values of $N$ under respective induced equilibria.}
    \label{fig:win-probability_apx}
\end{subfigure}
\hfill 
\begin{subfigure}{0.28\textwidth}

 \begin{tikzpicture}[scale=0.54]
    % Because we need to give the spy node a name to add the labels afterwards,
    % it is a bit more complicate than usual. To do so we need to `scope` the
    % spy. To avoid further error messages it seems we need to `scope` the whole
    % `axis` environment.
    \begin{scope}[
        % Give the spy options to the `scope`
        %spy using outlines={
         %   rectangle,
          %  magnification=10,
           % connect spies,
           % size=2cm,
           % blue,
        %},
    ]
        \begin{axis}[
            ylabel={\textbf{Equilibrium win probability}},
            xlabel={$\alpha$},
            xmin=0.5, xmax=1.2,
             ymin=0.5, ymax=1,
             label style={font=\Large},
                    tick label style={font=\Large}, 
            % (simplified this statement)
            xtick={0.5,  0.6,  0.7, 0.8,  0.9,   1,  1.1,  1.2}, legend pos=north west
            % (removed all unnecessary/unrelated stuff)
        ]

%\begin{axis}[
%    title={Temperature dependence of CuSO\(_4\cdot\)5H\(_2\)O solubility},
 %   xlabel={$\alpha$},
  %  ylabel={Equilibrium win probability},
  %  xmin=0.5, xmax=1.5,
  %  ymin=0.5, ymax=1,
  %  xtick={0.5,0.6,0.7,0.8,0.9,1, 1.1, 1.2, 1.3, 1.4},
  %  ytick={0,0.5,0.6,0.7,0.8,0.9},
  %  legend pos=north west,
  %  ymajorgrids=true,
  %  grid style=dashed,
%]
%fist plot
\addplot[ 
    color=blue,
    mark=square
    ]
            % (simplified the plot by removing a lot of coordinates and adding
            %  `smooth` to the options
%    x = [0.5, 0.625, 0.75, 0.875, 1, 1.25, 1.5, 1.75]
%y1 = [0.5, 0.541341, 0.58145,0.66935,0.95108, 1,1,1]
    coordinates {  (0.5,0.5034463993584801)(0.55, 0.5209968252423642)(0.6, 0.5374775010447751) (0.65, 0.572291459126275)( 0.7, 0.6087449472248506 ) (0.75, 0.6597175923323482)(0.8,0.7225977102031234) (0.85, 0.7822348675774462)(0.9,0.8469174102819855) (0.95,0.8989513770811633) (0.975,0.9235848422584841)(1,0.9429125196891085) (1.025, 0.9587758575127688) (1.05,0.9705707189361287)(1.1,0.9873102229031069) (1.15, 0.9954117338584801 ) (1.2,0.9986352165177431)
    };
%    [0.5034463993584801, 0.5209968252423642, 0.5374775010447751, 0.572291459126275, 0.6087449472248506, 0.6597175923323482, 0.7225977102031234, 0.7822348675774462, , , , , , , 

%second
    \addplot[
    color=cyan,
    mark=triangle
    ]
    coordinates {  (0.5, 0.5010963154212058) (0.55,0.5085490588469159) (0.6,0.5202529716971387)  (0.65, 0.5403472945607244) (0.7, 0.5702710110874876) (0.75, 0.6151564360550874)  (0.8, 0.6756813110811366) (0.85, 0.7469505587253285) (0.9,  0.8227145240083575)  (0.95, 0.8941930364247158) (0.975, 0.9241895067950894)  (1, 0.9492574754059184)  (1.025, 0.9681185863906605) (1.05, 0.9815032301789244)  (1.1, 0.9953040321681711)(1.15, 0.999242990990767)(1.2, 0.9999299920173971)
%[
    };
%third
    \addplot[
    color=red,
    mark=diamond
    ]
    coordinates { (0.5,0.5003468902059904) (0.55, 0.5029925155005879) (0.6, 0.509174210623359) (0.65, 0.521629510461562) (0.7, 0.5439623696839793)(0.75, 0.5807451241705572)  (0.8, 0.6348391717851638) ( 0.85, 0.7076024419027669) (0.9,  0.7954745493948241) (0.95, 0.8832932944912986) (0.975, 0.9207629392877437) (1,  0.9511620674243744)(1.025, 0.9732500282846483) (1.05, 0.9871896441097167)(1.1,  0.9982903515939576) (1.15,  0.9999105144690333) (1.2, 0.9999986813739195)
    
    %[ 
    };
%fourth
    \addplot[ 
    color=brown,
    mark=otimes
    ]
    coordinates {  (0.5, 0.5001097028971535) (0.55, 0.5012324305739424) (0.6, 0.5045321111698506)( 0.65,  0.5119818568953474)( 0.7, 0.5271480860811699)( 0.75,  0.5554127150201131)( 0.8,  0.6021369500337498)(0.85,  0.6728180753803086)(0.9, 0.7668542167736896)(0.95, 0.8696779614724737)(0.975,  0.9154547179317182)(1, 0.9518786838592211) (1.025,  0.9769308242369477)(1.05,  0.9910949188400563)(1.1, 0.9994568515423824)(1.15, 0.9999940908764602)(1.2, 0.9999999949959482)
    %[]
    };
    %fifth 
    \addplot[ 
    color=black,
    mark=*
    ]
    coordinates {  (0.5, 0.5000346912422677)(0.55, 0.5005458562520027)(0.6,  0.5020799539199186)(0.65, 0.506433665179634)(0.7, 0.5167118739393539)(0.75, 0.5378743422746576)(0.8,  0.5771469031991653)(0.85, 0.6427745084476345)(0.9,  0.7394354680494846)(0.95, 0.8554562683942735)(0.975, 0.9093048027818923)(1, 0.9521161311947801)(1.025, 0.979997219313373)(1.05, 0.9939073988915024)(1.1,  0.9998608347908574) (1.15, 0.999999832098464)(1.2, 0.9999999999983229)
    %[]
    };

            % crate a coordinate of the point you want to magnify
            %\coordinate (point) at (axis cs:0.755,0.712);
     %       \PointX = 1.25
      %      \PointY = 0.99963
%            \Getxycoords{point}{\PointX}{\PointY}
            % Get the coordinates of that point (to later use them)
%            \Getxycoords{point}{\PointX}{\PointY}

            % draw the dashed lines to the axis (using the defined coordinate)
%            \draw [red,dashed]
 %               (point -| {axis cs:\pgfkeysvalueof{/pgfplots/xmin},0})
  %                  -| ({axis cs:0,\pgfkeysvalueof{/pgfplots/ymin}} -| point);

            % unfortunately one cannot directly place the spy at an
            % axis coordinate, thus we define a `\coordinate` first
                %\coordinate (spy point) at (axis cs:0.76,0.53);
                %\Getxycoords{point}{\PointX}{\PointY}
            %\spy[color=red] on (point) in node (spy) at (spy point);
            \addlegendentry{$N = 10^4$}
\addlegendentry{$N = 10^5$}
\addlegendentry{$N = 10^6$}
\addlegendentry{$N = 10^7$}
\addlegendentry{$N = 10^8$}
        \end{axis}
    \end{scope}
\end{tikzpicture}

     %\caption{Win probability of  $A$ for $\beta = 0.5$ and different values of $\alpha$ in polynomial PPM model  for different values of  $N$. The trend reversal,  where win probability of $A$ for largest $N$ shifts from lowest to highest,  can be observed at   $\alpha =1$.   }
     \label{fig:fourth_apx}
\end{subfigure}
\caption{(L)  demonstrates election instance from Section~\ref{sec:simulation}. For  large value $N$, the pivot point $c^*$, two non-trivial equilibria $c_N^+$ and $c_N^-$, and the trivial equilibrium $c^0$ are shown. For $c^+_N$, the probability of a random voter to vote $A$ is proportional to $s_A(c^+_N)$. The $m(c^+_N)$ is   proportional to the margin of victory.
The bold arrows  indicate that $c^+_N, c^0_N$ are stable  equilibria whereas $c^-_N$ is not stable. (C) Win probability for different values of $N$ under respective induced equilibria. (R) Win probability of  $A$ for $\beta = 0.5$ and different values of $\alpha$ in polynomial PPM model  for different values of  $N$. The trend reversal can be observed at   $\alpha =1$.  \label{fig:sim_apx} 
}
\end{figure*}

As noted earlier in the main paper, the unique pivot point is $0.6$, so we expect the winning probability  to approach $\Phi(0.6^{-1})\cong 0.952$ in both $c^+_N$ and $c^-_N$ as $N$ increases.

Figure~\ref{fig:sim_apx} (C) demonstrates the Non-Jury Theorem, showing that the winning probability of A in the stable equilibrium $c^+$ is bounded away from 1, even as the total population size $N$ is increasing. Interestingly, the winning probability of B behaves differently and is \emph{decreasing} towards the $c^-$ equilibrium point. Hence  $B$ would always prefer a smaller fraction of the   population to vote whereas the popular candidate prefers a large  population to vote. Showing that the equilibrium $c^-$ is not stable.

We observe that the  win probabilities under $c^+$ and $c^-$ for respective winning candidates are significantly different. While $B$ wins with a high probability under $c^-$, $A$ wins under $c^+$ with a high probability for the same population size, for  a fixed population size, the win probabilities of these candidates under their favoured equilibria are not the same.

 Figure~\ref{fig:sim_apx} (R) shows the win probability of candidate $A$ under $c^+$ and  for different values of $\alpha$. 
We observe that with  higher values of $N$ equilibrium win probability of at $c^+$   approaches a step function around $\alpha=1=2\beta$. The larger value of $\alpha$ pushes the equilibrium   point $c^+$ towards the right, increasing the win probability of the popular candidate.  Yet this occurs slowly and   even for  high $N$ there is  a gradual increase of the win probability of $A$ with $\alpha$. 

\iffalse 
\paragraph{Partial turnout equilibrium}
Note that the pivot point is $c^*=0.6$, so we expect the winning probability  to approach $\Phi(0.6^{-1})\cong 0.952$ in both $c^+_N$ and $c^-_N$ as $N$ increases ( 
Fig.~4 (C)) demonstrates the Non-Jury Theorem for Example~\ref{example:four}, showing that the winning probability of A in the stable equilibrium $c^+$ is bounded by $\cong 0.955$, even as the  population size $N$ is increasing. 
Interestingly, the winning probability of $B$ behaves differently and  \emph{decreases} towards $c^*$ as the population grows 
 indicating that the left equilibrium approaches $c^*$ slower (so the margin remains high for longer) towards the $c^-$ equilibrium point. Hence  $B$ would always prefer a smaller fraction of the   population to vote whereas the popular candidate prefers a large  population to vote.
We observe that the  win probabilities under $c^+$ and $c^-$ for respective winning candidates are significantly different. While $B$ wins with a high probability under $c^-$, $A$ wins under $c^+$ with a high probability for the same population size.  Furthermore, for  a fixed population size, the win probabilities of these candidates under their favoured equilibria are not the same.

\paragraph{Varying the PPM parameters}
 Fig.~2 (R) shows  win probability of  $A$ under $c^+_N$ for different  values of $\alpha$ and $N$. Higher values of $N$  approach a step function around $\alpha=1$$=2\beta$. Yet this occurs slowly and   even for  high $N$ we observe a gradual increase of the win probability of $A$ with $\alpha$.
 \fi

\end{document}